\documentclass[11pt,reqno]{amsart}
\usepackage[utf8]{inputenc}
\usepackage[margin=2.5cm]{geometry}

\usepackage{amsmath}
\usepackage{amsfonts}
\usepackage{amssymb}
\usepackage{amsaddr}
\usepackage{amsthm}

\usepackage{parskip}

\usepackage{empheq}
\usepackage{float}
\usepackage{framed}

\usepackage{comment}
\usepackage{lmodern}

\usepackage{enumitem}

\usepackage{color}
\usepackage{graphicx}
\usepackage{caption}
\usepackage{subcaption}
\usepackage{placeins}
\usepackage{listings}

\usepackage{appendix}

\allowdisplaybreaks
\usepackage{bbm}
\usepackage{mathtools}

\newtheorem{theorem}{Theorem}[section]

\newtheorem{proposition}{Proposition}[section]
\newtheorem{claim}[theorem]{Claim}
\newtheorem{remark}[theorem]{Remark}

\usepackage[backend=bibtex,
style = numeric,
]{biblatex}
\title{Optimal Execution with identity optionality}
\date{\today}
\author{René Carmona*, Claire Zeng*}
\address{*Department of Operations Research and Financial Engineering \\
Princeton University \\
Princeton, NJ 08544}
\email{\texttt{\{}rcarmona, cszeng\texttt{\}} [at] princeton.edu}
\begin{document}
\begin{abstract}
This paper investigates the impact of anonymous trading on the agents' strategy in an optimal execution framework. It mainly explores the specificity of order attribution on the Toronto Stock Exchange, where brokers can choose to either trade with their own identity or under a generic anonymous code that is common to all the brokers. We formulate a stochastic differential game for the optimal execution problem of a population of $N$ brokers and incorporate permanent and temporary price impacts for both the identity-revealed and anonymous trading processes. We then formulate the limiting mean-field game of controls with common noise and obtain a solution in closed-form via the probablistic approach for the Almgren-Chris price impact framework. Finally, we perform a sensitivity analysis to explore the impact of the model parameters on the optimal strategy. 
\end{abstract}
 
\maketitle

\section{Introduction, motivations and literature review}

\indent With the advent of algorithmic trading and electronic markets, automated and high frequency trading have become an increasingly active area of research. As a result, a considerable amount of effort has been devoted to understanding market microstructure. The existing literature covers a wide variety of subjects, but three main categories can be broadly identified: 
\begin{itemize}
    \item Statistical arbitrage, or studying opportunities to make profits out of ``predictable'' returns (short-term alpha) or benefit from short-term inefficiencies in the market (as with pair-trading or futures-index arbitrage);
    \item Optimal execution, or determining the optimal schedule (for a given cost functional) in order to sell and/or acquire a large position in one or multiple assets while mitigating several risks (such as information leakage, price impact and adverse selection);
    \item Market making, or determining the optimal placement of limit orders to benefit by providing liquidity to markets. 
\end{itemize}

The study of market impact and the modelling of market frictions has been of prime interest in designing efficient trading strategies. This paper analyzes a multi-agent optimal execution problem, where many financial institutions and brokers must determine the strategy to liquidate or build a position on a specific asset while maximizing an expected profit objective function. 

Most of the early literature on optimal execution focuses on the single agent setting, where the trader is facing a trade-off between choosing a fast trading rate (to reach their goal as soon as possible to reduce the execution risk) and limiting their price impact (which pushes prices in an unfavourable direction on average). The initial framework for the single trader case is attributed to Almgren and Chriss, who consider both a permanent and an immediate price impact in \cite{Almgren_Chriss_1998}. \cite{Obizhaeva_Wang_2005} considers the same problem but with a transient price impact that has exponential decay, while \cite{Gatheral_2009} generalizes this approach by introducing a general decay kernel. We refer to \cite{Cartea_Jaimungal_Penalva_2015} for a detailed presentation on optimal execution with a more general objective function. 
 \cite{Predatory_Brunnermeier} and \cite{Carlin_Lobo_Viswanathan_2007} introduced the two-player setting, where one agent has a liquidation target and the other is trying to benefit from \textit{predatory trading} by exploiting this information. Considering many agents is essential to model the markets and is of particular interest to try to model some financial events. Indeed, there are some liquidity events that may force some traders to liquidate a large position of an asset within a relatively short time window. Changes in the membership in stock indexes such as the Russell 3000, is a case in point. ETFs are funds that track indexes; they try to minimize the tracking error by replicating the index of interest in their portfolio. As a result, any major change in the index composition causes the ETF to rebalance its portfolio, adding or dropping the same stocks as the index. For instance, the Russell US indexes undergo an annual reconstitution process and the benchmark composition is communicated in advance to the marketplace. This composition change is essential to make sure that the indexes reflect accurately the US equity market. At the end of May, the official modifications are announced, and will be effective at the end of June. June is therefore a transition month, during which ETFs and other institutions tracking the index must trade to rebalance their portfolio, so that it replicates the reconstitution portfolio by the end of June.

The first extensions to multi-agent settings modeled one large trader facing an exogenous order flow such as in \cite{Cartea_Jaimungal_2015}. Modeling the interaction in an endogenous way, that is, when the order flow and the price dynamics come from the interaction of the agents on the market, has been formulated in finite-player stochastic differential games and in mean field games. Developed initially in \cite{Lasry_Lions_2006a}, \cite{Lasry_Lions_2006b}, \cite{Lasry_Lions_2007} and in parallel in \cite{Huang_Malhame_Caines_2006}, \cite{Caines_Huang_Malhamé_2018}, mean field games have been applied to several problems in economics and finance: for instance, \cite{cardaliaguet2017mean} addresses the problem of crowd trading, where the traders interact through the asset mid-price process. The use of mean field games requires the assumption of symmetry among the agents, but heterogeneous preferences can be considered by introducing either a Major-Minor framework as in \cite{Jaimungal-MFG-Major-Minor} or several sub-populations. Several approaches are possible when solving a mean field game problem. The probabilistic approach, formulated by \cite{CarmonaDelarue1}, aims at characterizing directly optimal controls. This method has been applied to the optimal execution problem in several works starting with\cite{Carmona_Lacker_2015}. The variational approach, which relies on a Dynamic Programming Principle, has been used in \cite{cardaliaguet2017mean} and \cite{Jaimungal-MFG-Major-Minor}: it characterizes value functions directly, while incidentally characterizing optimal controls. Another approach based on applying convex analysis techniques can also be used as in \cite{Firoozi_Jaimungal_Caines_2020} and \cite{Neuman_Voss_2021}.  

With the introduction of pre-trade anonymity in equity markets, many exchanges have shifted towards a fully anonymous design; this has been further exacerbated by the increasing competition from Alternative Trading Systems (ATS) and Electronic Communication Networks (ECN). Publications addressing the impact of anonymity on market quality, in particular on liquidity, are few and far between. The consensus is that anonymity offers an additional opportunity for market participants to enhance their trading strategies for a better execution. Previous studies compared the effect of anonymity in different market designs. Some performed statistical analysis and comparison between separate platforms dedicated to anonymous and non-anonymous trading (\cite{Reiss_Werner_2005}) or before and after a regulatory change in identity disclosure requirements. The Toronto Stock Exchange (TSX) sets itself apart in that it has a hybrid system where anonymity and transparency co-exist side-by-side. The TSX public trading tape (Level I data) is also helpful in overcoming the obstacles mentioned above, both because it is one of the few markets where anonymity and identity are both actively used (at least for the most liquid stocks) and because it is quite significant in size; it represents the eleventh largest exchange world wide and is ranked third in North America in market capitalization. Voluntary identity disclosure is a key feature of the market design of the TSX \footnote{An additional feature of the Toronto stock exchange lies in the difference between disclosing one's identity or not. Indeed, anonymous orders are excluded from what is called \textit{broker preferencing}. This refers to the priority of the order matching on TSX: attributed orders will follow the Price/Broker/Long Life/Time priority to be matched (Long Life orders are committed to rest in the order book for a minimum period of time, during which they can be neither modified nor canceled). At the  best bid or offer price, attributed orders of one specific broker will be matched with new offsetting attributed orders from the same broker, and this, ahead of the other brokers that arrived before them. This  enables attributed orders to ``jump the queue'' and happens only if the identity is disclosed on both sides. Once these orders are matched, the Long Life/Time priority will be applied to the other orders. This broker preferencing allows for lower transaction costs, since internal crosses are free of fees. Therefore, not disclosing the identity on certain orders represents a potential cost, but this is not the object of our analysis.}. Although all agents are identified by a unique identity code (which is publicly available), every broker has the right to either disclose their identity or choose anonymity for each individual trade.  When anonymity is chosen, all the anonymous buyer/seller identities are reported under a generic anonymous code 01 both pre-trade and post-trade; otherwise, the exchange discloses the buyer and/or seller ID codes. Figure \ref{fig:inventory-example} below illustrates the inventory accumulated by the generic broker (representing the anonymous orders of all the brokers trading anonymously) and on the right the inventory of the specific broker with ID code \#65 for the stock RCI.B (Rogers Communications Inc. Class B) on 02/26/2022. The statistical analysis in  \cite{TSX_anonymous} tries to identify the determinants of the decision to trade anonymously. Based on a data set from the Market Regulation Services supervising the Canadian securities market, the authors find that reduced execution costs are associated with anonymous orders from strategic traders, while the trades displaying the identity of specialists and dual capacity brokers have a higher price impact. Moreover, they infer that anonymity is strategically selected, depending in particular on the market conditions but also on the order source, size, aggressiveness, and expected execution costs. 

\begin{figure}[ht]
\centering
\includegraphics[width=15cm, trim={1cm 2.5cm 0 1cm}]{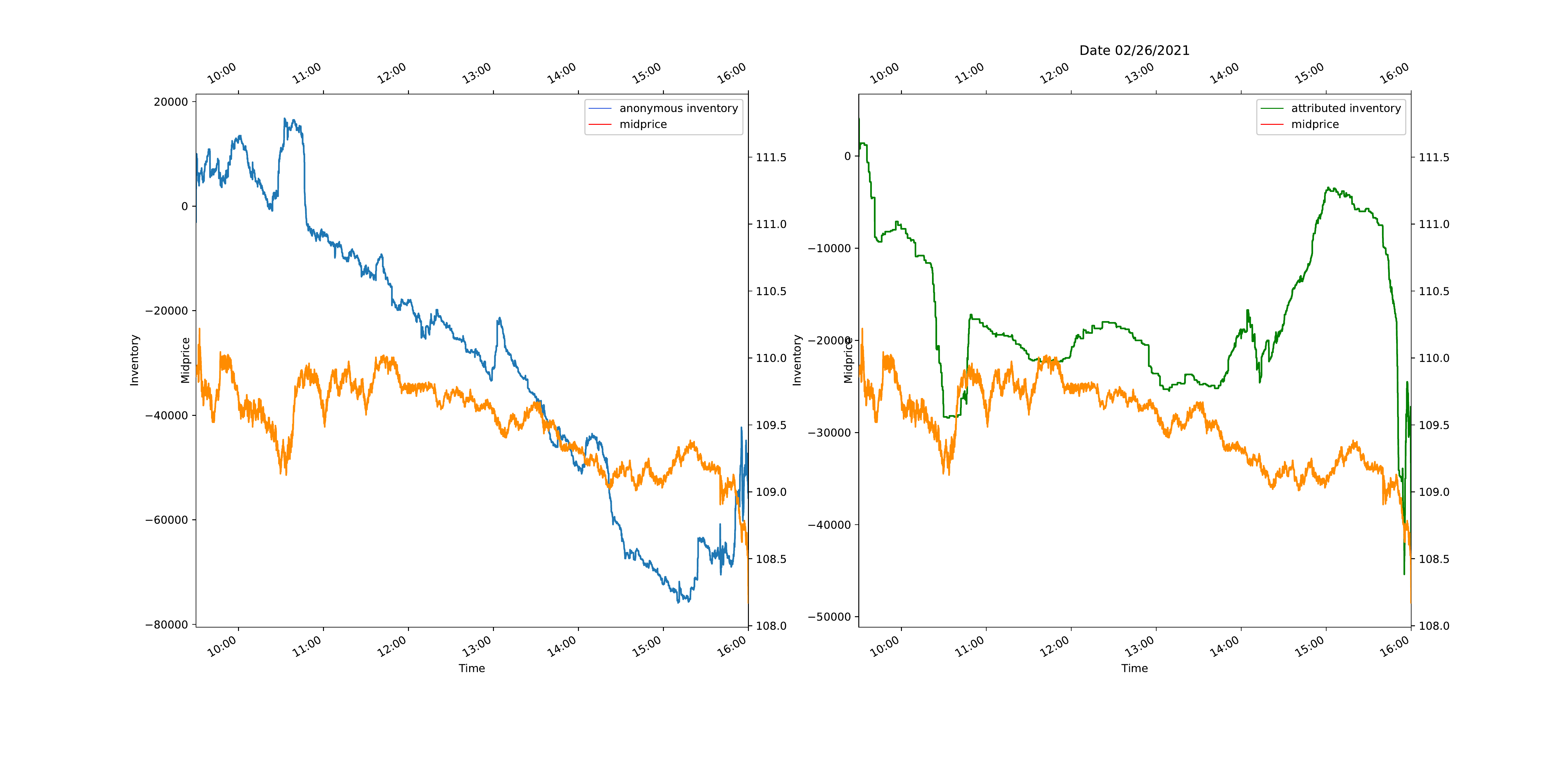}
\caption{Non-anonymous inventory for the anonymous broker \#1 (left) and broker \#65 (right) for the stock RCI.B on 02/26/2021}
\label{fig:inventory-example}
\end{figure}

This paper introduces an optimal execution stochastic differential game for a setting that takes into account identity optionality and whose limiting problem is a mean-field game of controls with common noise. Our analysis aims at determining if the voluntary disclosure of identity represents an opportunity for the market participants and investigates its impact on the agents' trading strategies, rather than on the market quality. Section \ref{section:finite} presents the general $N$-player stochastic differential game for the optimal execution problem with identity optionality. Section \ref{section:mfg} describes the general strategy to solve the mean field game limit. Section \ref{section:section_specific} specifies the model for an \textit{Almgren-Chriss}-type price impact model and develops the McKean-Vlasov Theory necessary to find an open-loop solution of the limiting mean field game formulated in the strong formulation and presents a numerical study. 
Finally, Section \ref{section:numerics} provides numerical illustrations and interpretations of the equilibrium characteristics for some realistic sets of parameters.
  
\vskip 6pt\noindent
\textbf{Notations:}
\vskip 2pt
We assume that we are given a complete probability space $(\Omega^0, \mathcal{F}^0, \mathbb{P}^0)$ endowed with a complete and right-continuous filtration $\mathbb{F}^0 = (\mathcal{F}^0_t)_{t \in [0,T]}$ generated by the Wiener process $\boldsymbol{W}^0 = (W^0_t)_{t \in [0,T]}$ and for each integer $i\ge 1$, a complete probability space $(\Omega^i, \mathcal{F}^i, \mathbb{P}^i)$ endowed with a complete and right-continuous filtration $\mathbb{F}^i = (\mathcal{F}^i_t)_{t \in [0,T]}$ generated by a two-dimensional Wiener process $\boldsymbol{W}^i = (\boldsymbol{W}^i_t)_{t \in [0,T]} = ((W^{i,1}_t,W^{i,2}_t)^{\dagger})_{t \in [0,T]}$.  The independent Wiener processes $\boldsymbol{W}^{i,1}$ and $\boldsymbol{W}^{i,2}$ play the roles of the idiosyncratic noises of player $i$ associated with the anonymous and non-anonymous trading processes respectively. We assume independence of the Brownian motions $\boldsymbol{W}^0, \boldsymbol{W}^1, \dots, \boldsymbol{W}^N, \dots$.

\section{Finite Player stochastic differential game}
\label{section:finite}
If $N\ge 1$ is an integer, when considering a model for $N$ players, we denote by $(\Omega^0\times\Omega^1\times\dots\Omega^N, \mathcal{F}^0 \otimes  \mathcal{F}^1  \otimes  \dots \otimes  \mathcal{F}^N, \mathbb{P}^0 \otimes \mathbb{P}^1 \otimes \dots \otimes \mathbb{P}^N)$ endowed with a complete and right-continuous filtration $\mathbb{F} = (\mathcal{F})_{t \in [0,T]}$ defined from the augmentation of the product filtration $\mathbb{F}^0 \otimes  \mathbb{F}^1  \otimes  \dots \otimes  \mathbb{F}^N$ so that it is right-continuous and complete. So in such a model, the common noise $W^0$ is essentially constructed on $(\Omega^0, \mathcal{F}^0, \mathbb{P}^0)$ whiile the idiosyncratic noises $(\boldsymbol{W}^i)_{i \geq 1}$ are essentially built on their respective $(\Omega^i, \mathcal{F}^i, \mathbb{P}^i)$.  

\subsection{Price Impact Model}
For simplicity we assume that all players trade the same stock whose price at time $t$ is denoted by $S_t$. The inventory $Q^i_t$ at time $t$ of broker $i$ is the aggregation of two separate inventories: an inventory $Q^{i,a}_t$ accumulated by trading anonymously at speed $\nu^{i,a}_t$ and an inventory $Q^{i,n}_t$ accumulated through identity-revealed trading at speed $\nu^{i,n}_t$. We ssume that these inventories have non-trivial quadratic variations and we denote by $\sigma_a$ and $\sigma_n > 0$ their volatilities, so: 
\begin{subequations} \label{eq:N_payer_inventories_dynamics} 
   \begin{alignat}{2}
      & d Q^{i,a}_t = \nu^{i,a}_t dt + \sigma_a dW^{i,1}_t \\
    & d Q^{i,n}_t = \nu^{i,n}_t dt + \sigma_n dW^{i,2}_t.
    \end{alignat} 
\end{subequations}
 
\begin{remark} Note the presence of  idiosyncratic Brownian motions in the dynamics of the inventory processes, and as we shall see later on, in the associated wealth processes derived from the self-financing conditions. The noise term $\sigma_a dW^{i,1}_t$ (resp. $\sigma_a dW^{i,1}_t$) models a random stream of client demands the broker faces. These demands  affect the broker anonymous (resp. identity revealed) inventory. This was first introduced and studied in \cite{carmonawebster2019applications}, and later empirically supported in \cite{carmonaleal2021optimal} where statistical tests performed on the Toronto Stock Exchange data show that both the inventory and the wealth dynamics should have non-zero quadratic variations. This assumption is also consistent with the option for a client to specify if they want anonymity or the broker identity to appear when they send their order. 
\end{remark}

Here, we use the modeling assumptions introduced in \cite{Carmona_Lacker_2015} and developed in \cite{carmonawebster2019applications} with a nonlinear order book. At each time $t$, every agent faces a cost structure given by two transaction cost curves $c_a , c_n:\mathbb{R} \mapsto [0,\infty]$, which are convex and satisfy $c_a(0) = c_n(0) = 0$. The order book incorporates each trade and reconstructs itself instantly around a new mid-price $S_t$, impacted in the following way by the transactions. If a single agent $i$ places an anonymous (resp. identity-revealed) market order of $\nu^{i,a}_t$ (resp. $\nu^{i,n}_t$) when the mid-price is $S_t$, the transaction will cost them $\nu^{i,a}_t S_t + c_a(\nu^{i,a}_t)$ (resp. $\nu^{i,n}_t S_t + c_n(\nu^{i,n}_t)$). Hence an anonymous and an identity-revealed trades will trigger changes in cash given by: 
\begin{subequations} \label{eq:N_payer_cash_processes_dynamics} 
   \begin{alignat}{2}
    & d K^{i,a}_t = - (\nu^{i,a}_t S_t + c_a(\nu^{i,a}_t))dt \\
    & d K^{i,n}_t = - (\nu^{i,n}_t S_t + c_n(\nu^{i,n}_t))dt
   \end{alignat} 
\end{subequations}
respectively .
Let us make the assumption that $c_a(\cdot) = \kappa_a c(\cdot)$ and $c_n = \kappa_n c(\cdot)$ for some positive constants $\kappa_a$ and $\kappa_n$, and a function $c : \mathbb{R} \mapsto \mathbb{R}$ which is convex and satisfies $c(0) = 0$. 

\vskip 1pt
Once a market orders are executed, the order book relocates around a price incorporating the \textit{permanent price impact}  composed of two terms :
\begin{itemize}\itemsep=-2pt
    \item the impact from the anonymous market orders : $\frac{\gamma_a}{N} \sum_{j=1}^N \kappa_a c'(\nu^{j,a}_t)$ 
    \item the impact from the identity-revealed market orders : $\frac{\gamma_n}{N} \sum_{j=1}^N \kappa_n c'(\nu^{j,n}_t) $ 
\end{itemize}
Hence the mid-price process can be modeled as a martingale plus a drift representing this permanent impact. 
\begin{equation} 
\label{eq:N_player_price_dynamics}
 dS_t = \Big( \frac{\gamma_a}{N} \sum_{j=1}^N \kappa_a c'(\nu^{j,a}_t) + \frac{\gamma_n}{N} \sum_{j=1}^N\kappa_n c'(\nu^{j,n}_t)  \Big) dt + \sigma_0 dW^0_t
\end{equation}
The wealth $V^{i,a}_t$ (resp.  $V^{i,n}_t$) accumulated by agent $i$ through anonymous (resp. attributed) trading at time $t$ will therefore be the sum of the initial value, the mark-to-market value of the anonymous (resp. attributed) inventory and the anonymous (resp. attributed) cash process:
\begin{align*}
   & V^{i,a}_t = V^{i,a}_0 + Q^{i,a}_t S_t + K^{i,a}_t \\
    & V^{i,n}_t = V^{i,n}_0 + Q^{i,n}_t S_t + K^{i,n}_t, 
\end{align*}
and since the Brownian motions $W^0_t$ is independent of $W^{i,1}_t$ and $W^{i,2}_t$, the wealth processes have the following dynamics:  
\begin{subequations} 
   \begin{alignat}{2}
    & d V^{i,a}_t = Q^{i,a}_t dS_t + S_t dQ^{i,a}_t + dK^{i,a}_t\\
    & d V^{i,n}_t = Q^{i,n}_t dS_t + S_t dQ^{i,n}_t + dK^{i,n}_t
   \end{alignat} 
\end{subequations}
We assume that the agents are \text{risk-neutral} and seek to maximize the expectation of their terminal wealth, a running cost $f$ and a terminal cost $g$ functions of their inventories. Agent $i$ therefore wants to maximize the following objective functional:
\begin{equation} 
\label{eq:N_player_objective_general}
J^{i}(\boldsymbol{\nu}^i, \boldsymbol{\nu}^{-i}) = \mathbb{E} \Big[ V^{i,a}_T + V^{i,n}_T + g(Q^{i,a}_T, Q^{i,n}_T) + \int_0^T f(Q^{i,a}_t, Q^{i,n}_t) dt \Big]  
\end{equation}
By using the expression of the wealth processes and discarding the constant terms, the objective functional to maximizes becomes: 
\begin{equation} \label{eq:N_player_objective_general_IBP}
J^{i}(\boldsymbol{\nu}^i, \boldsymbol{\nu}^{-i}) = \mathbb{E} \Big[ g(Q^{i,a}_T, Q^{i,n}_T) + (Q^{i,a}_T + Q^{i,n}_T)S_T + \int_0^T \Big( d(K^{i,a}_t + K^{i,n}_t) + f(Q^{i,a}_t, Q^{i,n}_t) dt \Big) \Big]  
\end{equation}
By using the expression of $\boldsymbol{K}_t$, the objective functional becomes : 
\begin{equation}
J^{i}(\boldsymbol{\nu}^i, \boldsymbol{\nu}^{-i}) = \mathbb{E} \Big[ g(Q^{i,a}_T, Q^{i,n}_T) + (Q^{i,a}_T + Q^{i,n}_T)S_T - \int_0^T \Big((\nu^{i,a}_t + \nu^{i,n}_t)S_t + c(\nu^{i,a}_t) + c(\nu^{i,n}_t) - f(Q^{i,a}_t, Q^{i,n}_t) \Big) dt \Big]  
\end{equation}
and we are using the various rates of trading as controls which we assume to be progressively measurable for the filtrations generated by the idiosyncratic noise and the common noise. 

\subsection{Hamiltonian}

At each given time, the state of the system is captured by the state variable $(s, \underline{\boldsymbol{q}} )$ where $s$ is the value of the mid-price, and  $\underline{\boldsymbol{q}}$ is an $N$-tuple $\underline{\boldsymbol{q}} = (\boldsymbol{q}^1, \dots, \boldsymbol{q}^N) \in (\mathbb{R}^{2})^N$ where $\boldsymbol{q}^{i} = (q^{i,a}, q^{i,n}) \in \mathbb{R}^2$ describes the private state of agent $i$. The first dual variable is $y^s$, the next one being $\underline{\boldsymbol{y}} = (\underline{\boldsymbol{y}}^{1}, \dots, \underline{\boldsymbol{y}}^{N}) \in \mathbb{R}^{2N}$ where each $\underline{\boldsymbol{y}}^{i}$ is itself an $N$-tuple $\underline{\boldsymbol{y}}^i = (\boldsymbol{y}^{i,1}, \dots, \boldsymbol{y}^{i,N}) \in \mathbb{R}^{2N}$, where  $\boldsymbol{y}^{i} = (y^{i,a}, y^{i,n})$. The last dual variables are $\boldsymbol{z}^s$ and $\underline{\boldsymbol{z}} = (\boldsymbol{z}^{1}, \dots, \boldsymbol{z}^{N}) \in (\mathbb{R}^2)^N$ where  $\boldsymbol{z}^{i} = (\boldsymbol{z^{i,s}}, \boldsymbol{z^{i,a}}, \boldsymbol{z^{i,n}})$ with $\boldsymbol{z^{i,l}} = (z^{i,l,0}, z^{i,l,1}, z^{i,l,2})$ for $l \in \{s,a,n\}$. 

\noindent For each $i \in \{1, \dots, N\}$, the Hamiltonian of player $i$ reads: 
\begin{align}
        H^i(s, \underline{\boldsymbol{q}}, y^s, \underline{\boldsymbol{y}}^i, z^s, \underline{\boldsymbol{z}}^i, \underline{\boldsymbol{\nu}}, \boldsymbol{\overline{\nu}}^{N}) = \, & \Big(\frac{\gamma_a  \kappa_a}{N}\sum_{j=1}^N c'(\nu^{j,a}_t) + \frac{\gamma_n \kappa_n}{N} \sum_{j=1}^N c'(\nu^{j,n}_t)\Big) y^s + \boldsymbol{\nu}^i \cdot \boldsymbol{y}^{i}  \nonumber \\ 
        & + \sigma z^s  +  \sigma^i_a z^{i,a,1} + \sigma^i_n z^{i,a,2} -  f(q^{i,a},q^{i,n})   \nonumber \\ 
        & + \kappa_a c(\nu^{i,a}) + \kappa_n c(\nu^{i,n}) + (\nu^{i,a} + \nu^{i,n} )s. 
\end{align} 

\subsection{Conditional Propagation of Chaos}
Let us denote the elements $\omega^0 \in \Omega^0$ and $\boldsymbol{\omega}^j = (\omega^{j,1}, \omega^{j,2}) \in \Omega^j$ such that $W^0_t(\omega) = \omega^0(t)$,  $W^{j,1}_t(\omega) = \omega^{j,1}(t)$  and $W^{j,2}_t(\omega) = \omega^{j,2}(t)$ for every $j \in \{1, \dots, N\}$. Generic elements of $\Omega$ are denoted $\omega = (\omega^0, \omega^{1,a}, \omega^{1,n}, \dots, \omega^{N,a},
\omega^{N,n}, \dots)$. 

\noindent We search for an equilibrium in \textbf{distributed open-loop} controls, i.e. the control of one individual depends only upon the history of their own idiosyncratic noises (the Brownian motions driving their own inventory processes) and the common noise (the Brownian motion driving the price process). We are therefore looking for open-loop policy functions $\boldsymbol{\pi} = (\pi^a, \pi^n)$ such that the actual controls used by trader $j$ at time $t \in [0,T]$ are of the form:
\begin{align}
    \nu^{N,j,a}_t(\omega) = \pi^a_t(\omega^0_{[0,t]}, \omega^{j,1}_{[0,t]},  \omega^{j,2}_{[0,t]}) \label{eq:policy_anon} \\
     \nu^{N,j,n}_t(\omega) = \pi^n_t(\omega^0_{[0,t]}, \omega^{j,1}_{[0,t]},  \omega^{j,2}_{[0,t]})  \label{eq:policy_id}
\end{align}
for deterministic progressively measurable (deterministic) functions 
$$
\pi^a, \pi^n : [0,T] \times \mathcal{C}([0,T], \mathbb{R}) \times \mathcal{C}([0,T], \mathbb{R}^2),
$$
where we use the notation $\mathcal{C}([0,T],\mathcal{X})$ for the space of continuous functions from $[0,T$ into $\mathcal{X}$. We then consider $(S^{N, \bar{\nu}}_t, Q^{N,i,a}_t, Q^{N,i,n}_t)$ as the private state of each trader $i$ by artificially duplicating the mid-price process $S^{N, \bar{\nu}}$. So the dynamics of the state of trader $i$ are given by : 
\begin{subequations}\label{eq:dynamics_N_state}
    \begin{alignat}{2}
    & dS^{N, \bar{\nu}}_t = \Big( \frac{\gamma_a \kappa_a}{N} \sum_{j=1}^N  c'(\nu^{N,j,a}_t)+ \frac{\gamma_n \kappa_n }{N} \sum_{j=1}^N c'(\nu^{N,j,n}_t) \Big) dt + \sigma_0 dW^0_t \label{eq:dynamics_N_midprice_state} \\
   & dQ^{N,i,a}_t =  \nu^{N,i,a}_t dt + \sigma_a dW^{i,1}_t \label{eq:dynamics_N_anoninventory_state} \\ 
   & dQ^{N,i,n}_t =  \nu^{N,i,n}_t dt + \sigma_n dW^{i,2}_t \label{eq:dynamics_N_idinventory_state}
   \end{alignat}
  \end{subequations}

\noindent Intuitively, in the limit $N \rightarrow \infty$, the equilibrium distribution of the controls should still feel the influence of the common noise $W^0$ and therefore, it should not be deterministic but rather $\mathbb{F}^0$-measurable. We conjecture that as $N$ tends to $\infty$, the empirical average of the $\{\nu^{N,i,a/n}_t)\}_{i \in [N]}$ converges towards the common conditional expectation of the respective cost of the controls given the common source of noise $W^0$. We therefore formulate the following problem. For each individual $i \in \{1, \dots, N\}$, we have : 
\begin{subequations}
    \begin{alignat}{2}
    & dS^{i, \bar{\nu}}_t = \Big( \gamma_a \kappa_a \mathbb{E}[c'(\nu^{j,a}_t)| \mathcal{F}^0_t] + \gamma_n \kappa_n \mathbb{E}[c'(\nu^{j,n}_t)| \mathcal{F}^0_t]\Big) dt + \sigma_0 dW^0_t, \quad S_0 = s_0  \label{eq:dynamics_limit_midprice_state} \\
   & dQ^{i,a}_t =  \nu^{i,a}_t dt + \sigma_a dW^{i,1}_t, \quad Q^{i,a}_0 = 0 \label{eq:dynamics_limit_anoninventory_state} \\ 
   & dQ^{i,n}_t =  \nu^{i,n}_t dt + \sigma_n dW^{i,2}_t , \quad Q^{i,n}_0 = 0 \label{eq:dynamics_limit_idinventory_state}
   \end{alignat}
  \end{subequations}
for $t \in [0,T]$ and where for $j \in \{1, \dots, N\}$ : 
\begin{align}
    \nu^{j,a}_t(\omega) = \pi^a_t(\omega^0_{[0,t]}, \omega^{j,1}_{[0,t]},  \omega^{j,2}_{[0,t]})  \\
     \nu^{j,a}_t(\omega) = \pi^n_t(\omega^0_{[0,t]}, \omega^{j,1}_{[0,t]},  \omega^{j,2}_{[0,t]}) 
\end{align}
\begin{claim}\label{claim_identical}
There exist progressively measurable functions $\phi^a, \phi^n : [0,+\infty) \times \Omega^0 \mapsto \mathbb{R}$ such that for all $\omega^0 \in \Omega^0, \, j \in \{1, \dots, N\}$: 
\begin{subequations}\label{eq:policy_identical}
    \begin{alignat}{2}
    & \phi^a(t, \omega^0) = \mathbb{E}[c'(\nu^{j,a}_t) | \mathcal{F}^0_t ](\omega^0)   \\
    & \phi^n(t, \omega^0) = \mathbb{E}[c'(\nu^{j,n}_t) | \mathcal{F}^0_t ](\omega^0)
   \end{alignat}
  \end{subequations}
\end{claim}

\begin{proof}
That's true, it is what was used for the first part but in a roundabout way. I changed the formulation to make it simpler: \\
By the assumption on the form of the controls \eqref{eq:policy_anon}-\eqref{eq:policy_id} and the definition of the conditional expectation, for every $j \in \{1, \dots, N\}$, there exist functions $\phi^{j,a}, \phi^{j,n} : [0,+\infty) \times \Omega^0 \mapsto \mathbb{R}$ such that \eqref{eq:policy_identical} holds. Since the controls $\nu^{j,a}$ (resp. $\nu^{j,n}$) involve the same open-loop policy $\pi^a$ (resp. $\pi^n$) and $(\boldsymbol{W}^i)_{i \geq 1}$ forms a sequence of independent Brownian motions, we conclude that the functions $\phi^{j,a}$ and $\phi^{j,n}$ are in fact the same across all $j \in \{1, \dots, N\}$. 
\end{proof}

From claim \ref{claim_identical}, we conclude that the processes $S^j_t$ are also the same across all $j \in \{1, \dots, N\}$. We can therefore write the dynamics of the state of trader $i \in \{1, \dots, N\}$ as : 
\begin{subequations}\label{eq:dynamics_limit_state}
    \begin{alignat}{2}
    & dS^{\bar{\nu}}_t = \Big( \gamma_a \kappa_a \mathbb{E}[c'(\nu^{a}_t)| \mathcal{F}^0] + \gamma_n \kappa_n \mathbb{E}[c'(\nu^{n}_t)| \mathcal{F}^0]\Big) dt + \sigma_0 dW^0_t, \quad S_0 = s_0  \label{eq:dynamics_limit_midprice_state_cv} \\
   & dQ^{i,a}_t =  \nu^{i,a}_t dt + \sigma_a dW^{i,1}_t, \quad Q^{i,a}_0 = 0 \label{eq:dynamics_limit_anoninventory_state_cv} \\ 
   & dQ^{i,n}_t =  \nu^{i,n}_t dt + \sigma_n dW^{i,2}_t , \quad Q^{i,n}_0 = 0 \label{eq:dynamics_limit_idinventory_state_cv}
   \end{alignat}
  \end{subequations}
for $t \in [0,T]$, where without loss of generality $\nu^a_t = \nu^{1,a}_t$ and $\nu^n_t = \nu^{1,a}_t$ for every $t \in [0,T]$. 

\begin{proposition}
Given the forms \eqref{eq:policy_anon} and \eqref{eq:policy_id}, it holds $\mathbb{P}$-almost surely: 
\begin{equation}
    \frac{1}{N} \sum_{j=1}^N c'(\nu^{j,a}_t) \underset{N \rightarrow \infty}{\longrightarrow} \mathbb{E}[c'(\nu^a_t) | \mathcal{F}^0_t],\, \quad \frac{1}{N} \sum_{j=1}^N c'(\nu^{j,n}_t) \underset{N \rightarrow \infty}{\longrightarrow} \mathbb{E}[c'(\nu^n_t)| \mathcal{F}^0_t]
\end{equation}
\end{proposition}
\begin{proof}
Let $t\in [0,T]$. 
Since $(\boldsymbol{W}^i)_{i \geq 1}$ is a sequence of independent Brownian motions, also independent of $W^0$, the $\{c'(\nu^{j,a}_t)\}_{j \geq 1}$ are conditionally independent and identically distributed given $\mathcal{F}^0$. \\
By the conditional Strong Law of Large Numbers, 
$$ 
\frac{1}{N} \sum_{j=1}^N c'(\nu^{j,a}_t) \underset{N \rightarrow \infty}{\longrightarrow} \mathbb{E}[c'(\nu^a_t)| \mathcal{F}^0],\qquad \mathbb{P} - a.s.
$$ 
The same argument holds for $\nu^{j,n}_t$.
\end{proof}

\section{Mean Field Game Formulation}
\label{section:mfg}

We now articulate and solve the mean field game formulation of the $N$-player game introduced above.
Before we proceed, we emphasize the fact that we still use the probability space $(\Omega^0, \mathcal{F}^0, \mathbb{P}^0)$ supporting the common noise, and we treat the \emph{generic player} as if it was the first player, de facto using only the probability space $(\Omega^1, \mathcal{F}^1, \mathbb{P}^1)$. We shall also use the two right-continuous and complete filtrations $\mathbb{F}^0 = (\mathcal{F}^0_t)_{0 \leq t \leq T}$ and $\mathbb{F}^1 = (\mathcal{F}^1_t)_{0 \leq t \leq T}$. Then, we define the product structure $\Omega = \Omega^0 \times \Omega^1, \, \mathcal{F}, \, \mathbb{F} = (\mathcal{F}_t)_{0 \leq t \leq T}, \mathbb{P}$ where $(\mathcal{F}, \mathbb{P})$ is the completion of $(\mathcal{F}^0 \otimes \mathcal{F}^1, \mathbb{P}^0 \otimes \mathbb{P}^1 )$ and $\mathbb{F}$ is the complete and right-continuous augmentation of $(\mathcal{F}^0_t \otimes \mathcal{F}^1_t)_{0 \leq t \leq T}$. 

Intuitively, in the limit $N \rightarrow \infty$, the equilibrium distribution of the controls should still feel the influence of the common noise $W^0$ and therefore,  they should not be deterministic but rather $\mathbb{F}^0$-measurable. This suggests that as $N$ tends to $\infty$, the empirical measures of  $\{\nu^{N,i,a}_t)\}_{i \in [N]}$ and of  $\{\nu^{N,i,n}_t)\}_{i \in [N]}$ converge towards a stochastic flow, that should be in equilibrium the conditional distribution of the optimal trading rate given the common source of noise $W^0$.

\noindent The search for a solution of the Mean Field Game will follow the classical strategy consisting in two steps : 
\begin{enumerate}
    \item For any arbitrary continuous $\mathbb{F}^0$-measurable stochastic process $\boldsymbol{\theta} = (\theta^a_t, \theta^n_t)_{0 \leq t \leq T}$, solve the optimization problem : 
    \begin{equation} \label{eq:generic_player_cost}
    \underset{\boldsymbol{\nu} \in \mathbb{A}}{\sup}  \,  J(\boldsymbol{\nu}, \boldsymbol{\theta})
    \end{equation}
    where 
    \begin{equation*}
    J(\boldsymbol{\nu}, \boldsymbol{\theta}) = \mathbb{E} \Big[ g(Q^{a}_T, Q^{n}_T) + (Q^{a}_T + Q^{n}_T)S_T - \int_0^T \Big( (\nu^{a}_t + \nu^{n}_t)S_t + \kappa_a c(\nu^{a}_t) + \kappa_n c(\nu^{i,n}_t) + f(Q^{a}_t, Q^{n}_t) \Big) dt \Big]  
    \end{equation*} 
subject to the dynamic constraint: 
\begin{subequations}\label{eq:dynamics_MFG_generic_state}
\begin{empheq}[left={\empheqlbrace\,}]{align}
    & dS_t = \Big( \gamma_a \kappa_a \theta^a_t + \gamma_n\kappa_n  \theta^n_t \Big) dt + \sigma_0 dW^0_t \label{eq:dynamics_MFG_generic_midprice_state} \\
   &  dQ^{a}_t = \nu^{a}_t dt + \sigma_a dW^1_t \label{eq:dynamics_MFG_generic_anoninventory_state} \\ 
   &  dQ^{n}_t = \nu^{n}_t dt + \sigma_n dW^2_t \label{eq:dynamics_MFG_generic_idinventory_state}
  \end{empheq}
  \end{subequations}
over the time interval $t \in [0,T]$, with $Q^a_0 = 0$, $Q^n_0 = 0$ and over controls which are adapted to both $\boldsymbol{W} = (W^1, W^2)$ and $W^0$. 

\item Determine the $\mathbb{F}^0$-measurable stochastic process $\boldsymbol{\theta} = (\boldsymbol{\theta}_t)_{0\leq t \leq T} = (\theta^a_t, \theta^n_t)_{0\leq t \leq T}$ so that the conditional marginal expectation of one optimal control with associated path $(\boldsymbol{Q}_t)_{0\leq t \leq T}$ given $W^0$ is precisely $(\boldsymbol{\theta}_t)_{0\leq t \leq T}$ itself, i.e. ;
\begin{equation}
    \forall t \in [0,T], \theta^a_t = \mathbb{E}\Big[c'(\nu^a_t) | \sigma \{ W^0_s,\, 0 \leq s \leq T \}\Big], \, \theta^n_t = \mathbb{E}\Big[c'(\nu^n_t) | \sigma \{ W^0_s,\, 0 \leq s \leq T \}\Big].
\end{equation}
\end{enumerate}

\subsection{Discussion of the main challenges}

This price impact model is particularly interesting because the interaction occurs through the (conditional) distribution of the controls and has only been studied in the literature for the case of no-identity optionality and without common noise. The price impact model was first formulated in \cite{Carmona_Lacker_2015} with no common noise, since the controls are adapted only to the idiosyncratic noise. The state process therefore reduces to the inventory. Via an application of Girsanov's theorem, they show the existence of weak solutions. \cite{CarmonaDelarue1} presents a slightly different model and solves it in the strong formulation. More generally, mean field games of controls without common noise have been the subject of substantial research. Indeed, the convergence of the value functions and Nash equilibria for the finite player games to the mean field limit counterpart has been addressed in \cite{Possamaï_Tangpi_2021} for models where the interaction occurs through the joint distribution of states and controls and where there is no common noise. A rate of convergence is derived for the particular Almgren-Chriss price impact model in a non-Markovian setting. \cite{Djete_2021} formulates the notion of measure-valued solutions and shows their existence for mean field games of controls in the no-common-noise setting. 
Some difficulties arise when we do not constrain the trading rates' adeptness only to the idiosyncratic noises. The presence of this additional noise and the application of the Pontryagin's Maximum Principle lead to a conditional McKean-Vlasov FBSDE system.  To the best of our knowledge, no general existence results are available for mean field games of controls with common noise.


\section{Explicit Solution of a Particular Model} 
\label{section:section_specific}

Given the difficulties mentioned above and for the sake of numerical analysis, we solve the model in the strong formulation for a model \textit{à la Cartea-Jaimungal}, namely one with an objective function that relies on the following price impact formulas: 
\begin{align} 
\label{eq:param_AC}
    \begin{cases}
   &  \gamma_a = \frac{\alpha_a}{2 \kappa_a}, \quad c_a(x) = \kappa_a x^2\\
   &  \gamma_n =  \frac{\alpha_n}{2 \kappa_n}, \quad c_n(x) = \kappa_n x^2 \\
   & f(t,q^a,q^n) = - \phi \lvert q^a + q^n \rvert ^2 \\
   & g(q^a, q^n) = - \Psi \lvert q^a + q^n - q_T \rvert ^2 
    \end{cases}
\end{align}
with the following admissibility set :
\begin{equation} \label{eq:admissible_set} 
    \mathbb{A}^{a/n} := \{ \nu \text{ is } \mathcal{F}\text{-predictable}, \, \nu \in \mathbb{H}^{2,1}_T \} 
\end{equation}
where $\mathbb{H}^{2,k}_T = \{ \nu : \Omega \times [0,T] \rightarrow \mathbb{R}^k, \mathbb{E}[\int_0^T \lvert \nu_u\rvert^2 du] < \infty\} $ for $k\geq 1$. 

\begin{remark}
Note that this choice for the price impacts corresponds to the initial Almgren-Chriss model. Indeed, the anonymous and identity-revealed temporary price impacts are parameterized by $\kappa_a >0$ and $\kappa_n>0$ respectively, while the permanent price impacts on the mid-price are given by $\alpha_a>0$ and $\alpha_n>0$. In order to be able to identify the relative contributions of the anonymous and the identity-revealed tradings, we assume that $\alpha_a \neq \alpha_n$ and $\kappa_a \neq \kappa_n$. The form of $g$ corresponds to a penalty on the difference between the actual terminal inventory and a chosen target $q_T$; it is parametrized by $\Psi$. The running cost corresponds to a penalization (parametrized by $\phi >0$) for inventory holding during the trading period. 
\end{remark}

\subsection{Solution in the Strong Formulation}
The optimization problem of a generic representative player has the form: 
  \begin{align}
    \underset{\boldsymbol{\nu} \in \mathbb{A}}{\sup}  \,  \mathbb{E} \Big[ & (Q^a_T + Q^n_T) S_T - \Psi (Q^{a}_T + Q^{n}_T - q_T)^2 - \phi \int_0^T (Q^{a}_t+Q^{n}_t)^2 dt     
\nonumber \\
    & -  \int_0^T \kappa_a (\nu^{a}_t)^2 dt  -  \int_0^T    \kappa_n (\nu^{n}_t)^2 dt  -  \int_0^T  (\nu^{a}_t + \nu^n_t ) S_t dt \Big]
    \label{eq:mfg-objective-function-4} 
    \end{align} 
  subject to the dynamic constraint: 
  \begin{align}
 \begin{cases}
         \label{eq:mfg-dynamics-4}
    dS_t = \Big( \alpha_a \mathbb{E}^0\Big[  \hat{\nu}^a_t  \Big] + \alpha_n  \mathbb{E}^0\Big[  \hat{\nu}^n_t  \Big] \Big) dt + \sigma_0 dW^0_t \\ 
         dQ^{a}_t = \nu^{a}_t dt + \sigma_a dW^1_t \\
         dQ^{n}_t = \nu^{n}_t dt + \sigma_n dW^2_t \\
    \end{cases}
\end{align}
over the time interval $t \in [0,T]$, with $Q^a_0 = 0$, $Q^n_0 = 0$, the optimization being over controls which are adapted to both $\boldsymbol{W}$ and $W^0$. Here, $\mathbb{E}^0$ stands for the conditional expectation given the filtration of the common noise.

\begin{proposition} 
\label{mfg-sol-proposition}
The optimal controls for the representative agent in the mean-field formulation \eqref{eq:mfg-objective-function-4}-\eqref{eq:mfg-dynamics-4}
can be written in feedback form as: 
\begin{subequations}
\label{eq:feedbackform_meanfield_trading_speeds}
\begin{alignat}{2}
    \hat{\nu}^a_t = -\frac{1}{2 \kappa_a} \Bigl[ \varphi_t (Q^a_t + Q^n_t)  + (\overline{\varphi}_t - \varphi_t) \mathbb{E}^0[{Q}^a_t + {Q}^n_t] + \overline{\chi}_t \Bigr] \\ 
    \hat{\nu}^n_t = -\frac{1}{2 \kappa_n} \Bigl[ \varphi_t (Q^a_t + Q^n_t) +  (\overline{\varphi}_t - \varphi_t) \mathbb{E}^0[{Q}^a_t + {Q}^n_t] + \overline{\chi}_t \Bigr]
\end{alignat}
\end{subequations}
where
\begin{align}
\begin{cases}
&  \varphi_t = \frac{-D \Bigl(e^{(\gamma^+ - \gamma^-)(T-t)} -1 \Bigr) - 2 \Psi \Bigl(\gamma^+ e^{(\gamma^+ - \gamma^-)(T-t)} - \gamma^-\Bigr)  }{\Bigl(\gamma^- e^{(\gamma^+ - \gamma^-)(T-t)} - \gamma^+\Bigr) - 2\Psi B \Bigl(  e^{(\gamma^+ - \gamma^-)(T-t)} - 1\Bigr) }, \quad t \in [0,T]\\
    & \overline{\varphi}_t = \frac{-C \Bigl(e^{(\delta^+ - \delta^-)(T-t)} -1 \Bigr) - 2 \Psi \Bigl(\delta^+ e^{(\delta^+ - \delta^-)(T-t)} - \delta^-\Bigr)  }{\Bigl(\delta^- e^{(\delta^+ - \delta^-)(T-t)} - \delta^+\Bigr) - 2\Psi B \Bigl(  e^{(\delta^+ - \delta^-)(T-t)} - 1\Bigr) }, \quad t \in [0,T] \\
    & \overline{\chi}_t = - 2 \Psi q_T \exp \Bigl\{ - \int_t^T \Bigl( \frac{\overline{\varphi}_s-\alpha_a}{2 \kappa_a}  + \frac{\overline{\varphi}_s-\alpha_n}{2 \kappa_n} \Bigr)ds \Bigr\}, \quad t \in [0,T]
\end{cases}
\end{align}
with $A = - \frac{1}{2} \Bigl( \frac{\alpha_a}{2 \kappa_a} + \frac{\alpha_n}{2 \kappa_n} \Bigr)$ and $B =   \frac{1}{2 \kappa_a}  + \frac{1}{2 \kappa_n}$, $C= 2\phi$, $\delta^{\pm} = A \pm \sqrt{A^2 + BC}$, $D=   \frac{\alpha_a}{2 \kappa_a}+ \frac{\alpha_n}{2 \kappa_n}$, $E =  2 \phi$ and $\gamma^{\pm} = \pm \sqrt{DE}$. 
\end{proposition} 

\subsection{Proof of Proposition \ref{mfg-sol-proposition}}
Let us fix a continuous $\mathbb{F}^0$-adapted stochastic process $\boldsymbol{\theta} = (\theta^a_t, \theta^n_t)_{0 \leq t \leq T}$ and solve the optimization problem of a generic representative player \eqref{eq:generic_player_cost} - \eqref{eq:dynamics_MFG_generic_state}. 

\noindent The Hamiltonian is given by : 
\begin{align}
    H(s,\boldsymbol{q}, y^s, \boldsymbol{y}^q, z^s, \boldsymbol{z}^q, \boldsymbol{\theta}, \boldsymbol{\nu}) = & (\alpha_a \theta^a + \alpha_n \theta^n) y^s + \nu^a y^a + \nu^n y^n + \kappa_a \lvert \nu^a \rvert^2 + \kappa_n \lvert \nu^n \rvert^2 + (\nu^a + \nu^n) s   \nonumber \\ 
    & + \phi \lvert q^a + q^n \rvert^2 + \sigma_0 z^s + \sigma_a z^{a,1} + \sigma_n z^{n,2} 
\end{align}  
and is uniquely minimized for: 
\begin{align}
    \hat{\boldsymbol{\nu}}  =   -\begin{pmatrix}
    \frac{1}{2 \kappa_a} & 0 \\
    0 &  \frac{1}{2 \kappa_n}
    \end{pmatrix} (s + \boldsymbol{y})
\end{align}
which is independent of the measure argument $\theta$ and where $\boldsymbol{y} = (y^a, y^n)^{\dagger}$. 

\noindent By the \emph{Pontryagin Stochastic Maximum Principle in Random Environment} (\cite{CarmonaDelarue1}) and using the consistency condition, we obtain the following conditional McKean-Vlasov FBSDE for \\
$(S,Q^a, Q^n,Y^s, Y^a, Y^n, \boldsymbol{Z^0},\boldsymbol{Z^1}, \boldsymbol{Z^2},\boldsymbol{M})$: 
\begin{subequations}
\begin{alignat}{2}
& dS_t = \Big( \alpha_a \mathbb{E}^0\Big[  \hat{\nu}^a_t  \Big] + \alpha_n  \mathbb{E}^0\Big[  \hat{\nu}^n_t  \Big] \Big) dt + \sigma_0 dW^0_t \\
    & dQ^a_t =  
    \hat{\nu}^{a}_t   dt +  
    \sigma_a  d W^a_t \\
      & dQ^n_t =  
    \hat{\nu}^{n}_t   dt +  
    \sigma_n  d W^n_t \\
    & dY^s_t = - (\hat{\nu}^a_t + \hat{\nu}^n_t) dt + Z^{s,0} dW^0_t + Z^{s,1}_t   dW^1_t   + Z^{s,2}_t   dW^2_t + dM^s_t \\ 
    & d Y^a_t =  - 2 \phi \Bigl( Q^a_t + Q^n_t \Bigr) dt 
     + Z^{a,0}_t   dW^0_t   + Z^{a,1}_t   dW^1_t   + Z^{a,2}_t   dW^2_t + dM^a_t\\
         & d Y^n_t =  - 2 \phi \Bigl( Q^a_t + Q^n_t \Bigr) dt 
     +Z^{n,0}_t   dW^0_t   + Z^{n,1}_t   dW^1_t   + Z^{n,2}_t   dW^2_t + dM^n_t\\
   &  Q^a_0 = Q^n_0 = 0, \, Y^s_T = - (Q^a_T + Q^n_T), \,  Y^a_T = Y^n_T = -S_T + 2\Psi(Q^{a} + Q^{n}_T - q_T) 
\end{alignat}
\end{subequations} 
where $M = (M^s, M^a, M^n)$ is a three-dimensional \textit{càd-làg} square-integrable martingale with respect to the filtration $\mathbb{F}$ with 0 as initial condition and of zero cross-variation with $W^0$, $(W^1,W^2)$. 

Choosing $Z_t^{s,0}= 0$, $ Z^{s,1} = -\sigma_a$, $Z_t^{s,2} = -\sigma_n $, and $M^s_t = 0$,  we see that $Y^s_t = - \Big( Q^a_t + Q^n_t \Big)$ solves the first adjoint equation. Let us make the following ansatz:  
\begin{align*}
    \begin{cases}
    Y^{a}_t = - S_t + v_a(t,Q^a_t + Q^n_t) \\ 
    Y^{n}_t = - S_t + v_n(t,Q^a_t + Q^n_t) 
    \end{cases}
\end{align*}
where the unknown functions $(t,q)\mapsto v_i(t,q)$, $i\in \{a,n\}$ satisfy the terminal condition $v_i(T,q)=2\Psi(q-q_T)$.
Then, the above McKean-Vlasov FBSDE system reduces to: 
\begin{align*}
\begin{cases}
    & Z^{a,0}_t = - \sigma_0, \quad Z^{a,1}_t = \sigma_a \partial_{q^a} v_a(t,Q^a_t+ Q^n_t), \quad Z^{a,2}_t = \sigma_a \partial_{q} v_a(t,Q^a_t+ Q^n_t), \quad M^a_t = 0 \\
    &  v_a \text{ s.t. }
     -2\phi q =  -(\alpha_a \mathbb{E}^0[\nu^{a}_t] + \alpha_n \mathbb{E}^0[\nu^{n}_t]) + \partial_t v_a(t,q) + \nu^a_t \partial_{q} v_a(t,q) + \frac{1}{2} \sigma_a^2 \partial^2_{q} v_a(t,q) \\
     & \qquad \qquad \qquad \quad \, +  \nu^n_t \partial_{q} v_a(t,q) + \frac{1}{2} \sigma_n^2 \partial^2_{q} v_a(t,q) \\
    & Z^{n,0}_t = - \sigma_0, \quad Z^{n,1}_t = \sigma_a \partial_{q} v_n(t,Q^a_t+ Q^n_t) , \quad Z^{n,2}_t = \sigma_n \partial_{q} v_n(t,Q^a_t+ Q^n_t), \quad M^n_t = 0  \\
    & v_n \text{ s.t. } 
     -2\phi  q =  -(\alpha_a \mathbb{E}^0[\nu^{a}_t] + \alpha_n \mathbb{E}^0[\nu^{n}_t]) + \partial_t v_n(t,q) + \nu^a_t \partial_{q} v_n(t,q) + \frac{1}{2} \sigma_a^2 \partial^2_{q} v_n(t,q)  \\
     & \qquad \qquad \qquad \quad \, +  \nu^n_t \partial_{q} v_n(t,q) +  \frac{1}{2} \sigma_n^2 \partial^2_{q} v_n(t,q),
\end{cases}
\end{align*}
and using the ansatz in the expressions for $\hat{\nu}$, we obtain : 
\begin{subequations}
\begin{alignat}{2}
      &\hat{\nu}^a_t =  - \frac{ 1}{2\kappa_a} v_a(t, Q^a_t+ Q^n_t)  \\
     &\hat{\nu}^n_t =   - \frac{ 1}{2\kappa_n} v_n(t,  Q^a_t +Q^n_t) 
\end{alignat}
\end{subequations} 
 Plugging these expressions into the partial differential equations for $v_a(t,q) $ and  $v_n(t,q)$, we obtain : 
\begin{align*}
\begin{cases}
\partial_t v_a(t,q) +\frac{\sigma_a^2+\sigma_n^2}{2}\partial^2_q v_a(t,q) - \frac12\Bigl(\frac{1}{\kappa_a}v_a(t,q)+ \frac{1}{\kappa_n}v_n(t,q)\Bigr) & \partial_q v_a(t,q)
+ 2\phi q \\
& - (\alpha_a \mathbb{E}^0[\hat{\nu}^{a}_t] + \alpha_n \mathbb{E}^0[\hat{\nu}^{n}_t]) =0 \\ 
  \partial_t v_n(t,q) +\frac{\sigma_a^2+\sigma_n^2}{2}\partial^2_q v_n(t,q)  - \frac12\Bigl(\frac{1}{\kappa_n}v_n(t,q)+ \frac{1}{\kappa_a}v_a(t,q)\Bigr) & \partial_q v_n(t,q) + 2\phi q \\
 & - (\alpha_a \mathbb{E}^0[\hat{\nu}^{a}_t] + \alpha_n \mathbb{E}^0[\hat{\nu}^{n}_t]) = 0
   \end{cases}
   \end{align*}
with the same terminal conditions $v_a(T, q) = 2\Psi(q - q_T)$ and $v_n(T, q) = 2\Psi(q - q_T)$. 

\vskip 2pt
We pursue the search for $v_a(t,q) $ and  $v_n(t,q)$ by looking for solutions in the forms $v_a(t,q) = \varphi_t q + \chi_t $ and $v_n(t,q) = \zeta_t q + \eta_t $. We first compute the conditional mean functions $\overline{Q}^a_t = \mathbb{E}^0[Q^{a}_t]$, $\overline{Q}^n_t = \mathbb{E}^0[Q^{n}_t]$, $\overline{v}^a(t) = \mathbb{E}^0[v^{a}(t, Q^a_t + Q^n_t)]$ and $\overline{v}^n(t) = \mathbb{E}^0[v^{n}(t, Q^a_t + Q^n_t)]$. The ansatz leads to $\overline{v}^a(t) = \overline{\varphi}_t \mathbb{E}^0[Q^a_t + Q^n_t] + \overline{\chi}_t$ and $\overline{v}^n(t) = \overline{\zeta}_t\mathbb{E}^0[Q^a_t + Q^n_t] + \overline{\eta}_t$. 

\noindent We obtain the following system of ordinary differential equations: 
\begin{subequations}\label{eq:adjoint_ansatz_mf_fixed}
    \begin{alignat}{2}
   &  \dot{\overline{\varphi}}_t - \frac{1}{2 \kappa_a} \overline{\varphi}_t^2 - \frac{1}{2\kappa_n} \overline{\zeta}_t \overline{\varphi}_t +  \frac{\alpha_a}{2 \kappa_a}  \overline{\varphi}_t  + \frac{
   \alpha_n}{2\kappa_n} \overline{\zeta}_t +  2 \phi  = 0, \quad \overline{\varphi}_T = 2 \Psi\\
   & \dot{\overline{\chi}}_t = \frac{\overline{\varphi}_t-\alpha_a}{2 \kappa_a}   \overline{\chi}_t + \frac{\overline{\varphi}_t-\alpha_n}{2 \kappa_n}  \overline{\eta}_t , \quad \overline{\chi}_T = -2 \Psi q_T \label{eq:adjoint_ansatz_mf_fixed_chi}\\ 
    &  \dot{\overline{\zeta}}_t  - \frac{1}{2 \kappa_n} \overline{\zeta}_t^2 - \frac{1}{2\kappa_a} \overline{\zeta}_t \overline{\varphi}_t + \frac{\alpha_n}{2 \kappa_n}   \overline{\zeta}_t + \frac{
   \alpha_a}{2\kappa_a} \overline{\varphi}_t +  2 \phi = 0, \quad \overline{\zeta}_T = 2 \Psi\\
   & \dot{\overline{\eta}}_t  = \frac{\overline{\zeta}_t-\alpha_a}{2 \kappa_a}   \overline{\chi}_t + \frac{\overline{\zeta}_t-\alpha_n}{2 \kappa_n}  \overline{\eta}_t , \quad \overline{\eta}_T = -2 \Psi q_T \label{eq:adjoint_ansatz_mf_fixed_eta}
    \end{alignat}
 \end{subequations}

\noindent Note that because of the symmetry in the Riccati equations, we can search for a solution for which $\varphi_t=\zeta_t$, in which case both equations become 
\begin{equation}
\dot{\overline{\varphi}}_t = 2A\overline{\varphi}_t + B \overline{\varphi}_t^2 - C,\qquad \overline{\varphi}_T=2\Psi
\end{equation}
with $A = - \frac{1}{2} \Bigl( \frac{\alpha_a}{2 \kappa_a} + \frac{\alpha_n}{2 \kappa_n} \Bigr)$ and $B =   \frac{1}{2 \kappa_a}  + \frac{1}{2 \kappa_n}$ and $C= 2\phi$. Since $BC>0$, this scalar Riccati equation has a unique non-exploding solution given by : 
\begin{equation}
    \overline{\varphi}_t = \frac{-C \Bigl(e^{(\delta^+ - \delta^-)(T-t)} -1 \Bigr) - 2 \Psi \Bigl(\delta^+ e^{(\delta^+ - \delta^-)(T-t)} - \delta^-\Bigr)  }{\Bigl(\delta^- e^{(\delta^+ - \delta^-)(T-t)} - \delta^+\Bigr) - 2\Psi B \Bigl(  e^{(\delta^+ - \delta^-)(T-t)} - 1\Bigr) }, \quad t \in [0,T]
\end{equation}
where $\delta^{\pm} = A \pm \sqrt{R}$ with $R = A^2 + BC  > 0$. 
Now that we have a unique solution $(\overline{\varphi}_t, \overline{\zeta}_t)$ on $[0, T]$, we look at the two other equations from the system \eqref{eq:adjoint_ansatz_mf_fixed}. Note that by the symmetry of the equations \eqref{eq:adjoint_ansatz_mf_fixed_chi} and \eqref{eq:adjoint_ansatz_mf_fixed_eta}, we can assume $\overline{\chi} = \overline{\eta}$. We end up with the following linear ordinary differential equation : 
\begin{equation}
    \dot{\overline{\chi}}_t = \Bigl( \frac{\overline{\varphi}_t-\alpha_a}{2 \kappa_a}  + \frac{\overline{\varphi}_t-\alpha_n}{2 \kappa_n} \Bigr) \overline{\chi}_t, \quad \overline{\chi}_T = -2 \Psi q_T
\end{equation}

\noindent This equation has the solution :
\begin{align}
    \overline{\chi}_t = - 2 \Psi q_T \exp \Bigl\{ - \int_t^T \Bigl( \frac{\overline{\varphi}_s-\alpha_a}{2 \kappa_a}  + \frac{\overline{\varphi}_s-\alpha_n}{2 \kappa_n} \Bigr)ds \Bigr\}
\end{align}

\noindent If we introduce the notation $\overline{V}_t =  \overline{Q}^a_t + \overline{Q}^n_t$, by summing the stochastic differential equations satisfied by $Q^a_t$ and $Q^n_t$, and by taking conditional expectations with respect to $\mathbb{F}^0$ we get: 
\begin{equation}
    d\overline{V}_t = -\Bigl(\frac{1}{2 \kappa_a} + \frac{1}{2 \kappa_n} \Bigr)\Bigl( \overline{\varphi}_t \overline{V}_t + \overline{\chi}_t  \Bigr) dt 
\end{equation}

\noindent If $\varphi_t \neq 0, \, t \in [0,T]$, then we obtain a partial differential equation of the form : 
\begin{equation}
    d\overline{V}_t = \theta(t) (f(t) - \overline{V}_t) dt
\end{equation}
where $\theta(t) = \Bigl(\frac{1}{2 \kappa_a} + \frac{1}{2 \kappa_n} \Bigr)  \overline{\varphi}_t $, $f(t) =  -\overline{\varphi}^{-1}_t \overline{\chi}_t $ and the solution is : 

\begin{equation}
\overline{Q}^a_t + \overline{Q}^n_t = \overline{V}_t = \int_0^t  e^{-\int_s^t \theta(r) dr} \theta(s) f(s) ds 
\label{eq:expected_inventory_sum_OU_solution}
\end{equation}

\noindent We can now go back to the original McKean-Vlasov FBSDE and solve for the equilibrium processes $\boldsymbol{Q}$ and $\boldsymbol{Y}$. Using the ansatz $v_a(t,q) = \varphi_t q + \chi_t $ and $v_n(t,q) = \zeta_t q + \eta_t $ and plugging the expression of $\overline{{\boldsymbol{Y}}}$, we obtain the system of ordinary differential equations : 
\begin{subequations}\label{eq:adjoint_ansatz_mf_fixed_original}
    \begin{alignat}{2}
   &  \dot{{\varphi}}_t - \frac{1}{2 \kappa_a} {\varphi}_t^2 - \frac{1}{2\kappa_n} {\zeta}_t {\varphi}_t  +  2 \phi  = 0, \quad {\varphi}_T = 2 \Psi\\
   & \dot{{\chi}}_t = \frac{{\varphi}_t}{2 \kappa_a}   \chi_t + \frac{{\varphi}_t}{2 \kappa_n}  {\eta}_t - \frac{\alpha_a}{2 \kappa_a} \overline{v}^a_t - \frac{\alpha_n}{2 \kappa_n} \overline{v}^n_t  , \quad {\chi}_T = -2 \Psi q_T \label{eq:adjoint_ansatz_mf_fixed_chi_orig}\\ 
    &  \dot{{\zeta}}_t  - \frac{1}{2 \kappa_n} {\zeta}_t^2 - \frac{1}{2\kappa_a} {\zeta}_t {\varphi}_t +  2 \phi = 0, \quad {\zeta}_T = 2 \Psi\\
   & \dot{{\eta}}_t  = \frac{{\zeta}_t}{2 \kappa_a}   {\chi}_t + \frac{{\zeta}_t}{2 \kappa_n}  {\eta}_t - \frac{\alpha_a}{2 \kappa_a} \overline{v}^a_t - \frac{\alpha_n}{2 \kappa_n} \overline{v}^n_t  , \quad {\eta}_T = -2 \Psi q_T .\label{eq:adjoint_ansatz_mf_fixed_eta_orig}
    \end{alignat}
 \end{subequations}

\noindent The first and third equations form a system of coupled Riccati equations. Again, by the symmetry, we deduce that $\varphi = \zeta$ and they are given by : 
\begin{equation}
    \varphi_t = \frac{-D \Bigl(e^{(\gamma^+ - \gamma^-)(T-t)} -1 \Bigr) - 2 \Psi \Bigl(\gamma^+ e^{(\gamma^+ - \gamma^-)(T-t)} - \gamma^-\Bigr)  }{\Bigl(\gamma^- e^{(\gamma^+ - \gamma^-)(T-t)} - \gamma^+\Bigr) - 2\Psi B \Bigl(  e^{(\gamma^+ - \gamma^-)(T-t)} - 1\Bigr) }, \quad t \in [0,T]
\end{equation}
where $\gamma^{\pm} = \pm \sqrt{S}$ with $S = DE > 0$,  $D=   \frac{\alpha_a}{2 \kappa_a}+ \frac{\alpha_n}{2 \kappa_n}$ and $E =  2 \phi$.
By the same symmetry argument, we can assume that $\chi = \eta$ and solve : 

\begin{equation}
     \dot{{\chi}}_t = \frac{{\varphi}_t}{2 \kappa_a}   \chi_t + \frac{{\varphi}_t}{2 \kappa_n}  \chi_t - \frac{\alpha_a}{2 \kappa_a} \overline{v}^a_t - \frac{\alpha_n}{2 \kappa_n} \overline{v}^n_t  , \quad {\chi}_T = -2 \Psi q_T
\end{equation}
by injecting the expression of $\chi$. The solution is then : 
\begin{equation}
    \chi_t = \Bigl( \frac{\alpha_a}{2 \kappa_a} + \frac{\alpha_n}{2 \kappa_n} \Bigr) \int_t^T \overline{v}_s e^{- \Bigl( \frac{\alpha_a}{2 \kappa_a} + \frac{\alpha_n}{2 \kappa_n} \Bigr) \int_t^s \varphi_r dr}ds, \, t \in [0,T]
\end{equation}
where $\overline{v} = \overline{v}^a = \overline{v}^n = \overline{\varphi} \Bigl(\overline{Q}^a + \overline{Q}^n \Bigr) + \overline{\chi}$.

\noindent Since $v^a_t = v^n_t = \varphi_t (Q^a_t + Q^n_t) + \chi_t$, we have $\overline{v}_t = \varphi_t (\overline{Q}^a_t + \overline{Q}^n_t) + \chi_t$, so $\chi_t = (\overline{\varphi}_t - \varphi_t) (\overline{Q}^a_t + \overline{Q}^n_t) + \overline{\chi}_t$. Hence, the optimal controls also take the form : 
\begin{subequations}
\label{eq:feedbackform_meanfield_trading_speeds}
\begin{alignat}{2}
    \hat{\nu}^a_t = -\frac{1}{2 \kappa_a} \Bigl[ \varphi_t (Q^a_t + Q^n_t)  + (\overline{\varphi}_t - \varphi_t) (\overline{Q}^a_t + \overline{Q}^n_t) + \overline{\chi}_t \Bigr] \\ 
    \hat{\nu}^n_t = -\frac{1}{2 \kappa_n} \Bigl[ \varphi_t (Q^a_t + Q^n_t) +  (\overline{\varphi}_t - \varphi_t) (\overline{Q}^a_t + \overline{Q}^n_t) + \overline{\chi}_t \Bigr]
\end{alignat}
\end{subequations}
These expressions are in the form of feedback functions of the current total inventory of the generic trader and the conditional expectations of the total inventories of the whole trading population.

\noindent By denoting $V_t =  Q^a_t + Q^n_t$, we obtain by summation: 
\begin{equation}
    dV_t = -\Bigl(\frac{1}{2 \kappa_a} + \frac{1}{2 \kappa_n} \Bigr)\Bigl( \varphi_t V_t +  (\overline{\varphi}_t - \varphi_t) (\overline{Q}^a_t + \overline{Q}^n_t) + \overline{\chi}_t\Bigr) dt +    \sigma_a  d W^a_t  + \sigma_n  d W^n_t 
\end{equation}

\noindent If $\varphi_t \neq 0, \, t \in [0,T]$, then the dynamics of $V_t$ are given by the stochastic differential equation: 
\begin{equation}
    dV_t = \beta(t) (g(t) - V_t) dt + \sigma_a dW^a_t + \sigma_n dW^n_t
\end{equation}
where $\beta(t) = \Bigl(\frac{1}{2 \kappa_a} + \frac{1}{2 \kappa_n} \Bigr)  \varphi_t $, $g(t) =  -\varphi^{-1}_t \Bigl( (\overline{\varphi}_t - \varphi_t) (\overline{Q}^a_t + \overline{Q}^n_t) + \overline{\chi}_t\Bigr) $ and the solution is : 
\begin{equation}
{Q}^a_t + {Q}^n_t = V_t = \int_0^t e^{-\int_s^t \beta(r) dr} \beta(s) g(s) ds + \int_0^t \sigma_a e^{-\int_s^t \beta(r) dr} dW^a_s + \int_0^t \sigma_n e^{-\int_s^t \beta(r) dr} dW^n_s 
\end{equation}
From the expressions of $\overline{Q}^a_t + \overline{Q}^n_t$ and ${Q}^a_t + {Q}^n_t$, we obtain the following linear feedback-form expressions of the optimal controls: 
\begin{align}
    \hat{\nu}^a_t = -\frac{1}{2 \kappa_a} \Bigl[ & \varphi_t \Bigl(\int_0^t e^{-\int_s^t \beta(r) dr} \beta(s) g(s) ds + \int_0^t \sigma_a e^{-\int_s^t \beta(r) dr} dW^a_s + \int_0^t \sigma_n e^{-\int_s^t \beta(r) dr} dW^n_s \Bigr) \nonumber \\
    & + (\overline{\varphi}_t - \varphi_t) \int_0^t  e^{-\int_s^t \beta(r) dr} \beta(s) f(s) ds + \overline{\chi}_t \Bigr] \\ 
    \hat{\nu}^n_t = -\frac{1}{2 \kappa_n} \Bigl[& \varphi_t \Bigl(\int_0^t e^{-\int_s^t \beta(r) dr} \beta(s) g(s) ds + \int_0^t \sigma_a e^{-\int_s^t \beta(r) dr} dW^a_s + \int_0^t \sigma_n e^{-\int_s^t \beta(r) dr} dW^n_s \Bigr) \nonumber \\
    & + (\overline{\varphi}_t - \varphi_t) \int_0^t  e^{-\int_s^t \beta(r) dr} \beta(s) f(s) ds + \overline{\chi}_t \Bigr] 
\end{align}
Therefore, the difference in the trading speeds of the anonymous or identity attributed inventories stems from the ratio $\frac{\kappa_a}{\kappa_n}$. This will be illustrated in the next section. 

\subsection{$\epsilon$-Nash Equilibrium} 
The goal of this subsection is to show that the solution of the limiting mean field game can be used to provide approximate Nash equilibria for the finite player versions of the model. This type of result is expected in the theory of mean field games so we defer its proof to Appendix \ref{Appendix_A}.

\begin{theorem} \label{thm:epsilon-NE} 
For the finite-population stochastic differential game defined in \eqref{eq:dynamics_N_state} for the choice \eqref{eq:param_AC}. of coefficients and parameters, if we consider the set of controls indexed by 
$j \in \{1,\dots, N\}$ and defined by: 
\begin{align}
    \hat{\nu}^{j,a}_t = -\frac{1}{2 \kappa_a} \Bigl[ \varphi_t (Q^{j,a}_t + Q^{j,n}_t)  + (\overline{\varphi}_t - \varphi_t) (\overline{Q}^a_t + \overline{Q}^n_t) + \overline{\chi}_t \Bigr] \\ 
    \hat{\nu}^{j,n}_t = -\frac{1}{2 \kappa_n} \Bigl[ \varphi_t (Q^{j,a}_t + Q^{j,n}_t)  +  (\overline{\varphi}_t - \varphi_t) (\overline{Q}^a_t + \overline{Q}^n_t) + \overline{\chi}_t \Bigr] 
\end{align}
where $(\overline{Q}^a_t + \overline{Q}^n_t)$ is defined in \eqref{eq:expected_inventory_sum_OU_solution}, then this set of controls provides an $\epsilon$-Nash equilibrium of the finite player stochastic differential game, in the sense: 
\begin{align}
    J^j(\hat{\nu}^j, \hat{\nu}^{-j}) \leq \underset{\boldsymbol{\nu} \in \mathcal{A}^j}{\sup} J^j(\nu, \hat{\nu}^{-j}) \leq J^j(\hat{\nu}^j, \hat{\nu}^{-j}) + O(\frac{1}{N})
\end{align}
\end{theorem}

\section{Numerical Results} 
\label{section:numerics}

This section provides numerical illustrations of the properties of the Nash equilibrium for the mean field game solved in the previous section.

\vskip 2pt
The closed-form expressions of the total inventory of the representative player, the conditional expectation of the total inventory of the mean field population of traders and the trading speeds in feedback form  \eqref{eq:feedbackform_meanfield_trading_speeds} enable us to compute the optimal trading speeds. We proceed to do just that for a given set of parameters. Next, we shall illustrate the influence of the ratio $\frac{\kappa_a}{\kappa_n}$. 

\vskip 2pt
For the purpose of our numerical computations we use the following set of parameters. We shall emphasize the instances in which some of these values will be modified. 

\begin{equation}
\boxed{
\begin{array}{rcl}
& \alpha_a = 4\times10^{-3}, \alpha_n = 5\times10^{-3}  & \text{permanent price impacts}\\
& \kappa_a = 1.5\times10^{-3}, \kappa_n = 3\times10^{-3}  & \text{temporary price impacts}\\
& \sigma_a = 2, \sigma_n = 4  & \text{volatilities}\\
& \Psi = 1, \phi = 1 & \text{terminal penalization and risk aversion} \\
& q_T = 200, T = 1.0   & \text{target inventory} \text{ and time horizon} \\

\end{array}
}
\label{eq:params_simulation_solo_set_0}
\end{equation}

\begin{figure}[H]
    \centering
    \begin{subfigure}[t]{0.45\textwidth}
        \centering
        \includegraphics[width = 7cm]{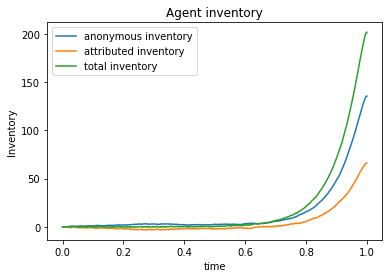}
        \caption{Optimal trading speeds}
            \label{fig:agent_trading_speeds_set_0}
    \end{subfigure}%
    ~ 
    \begin{subfigure}[t]{0.45\textwidth}
        \centering
          \includegraphics[width = 7cm]{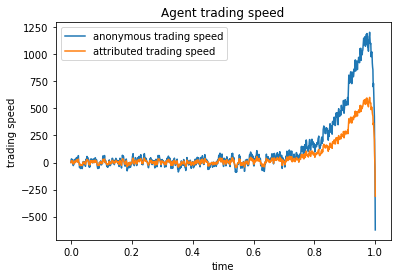}
    \caption{Optimal inventories}
    \label{fig:agent_inventories_set_0}
    \end{subfigure}
    \caption{Inventory and trading speed of the representative agent for the set of parameters \eqref{eq:params_simulation_solo_set_0}.}
\end{figure}

\noindent Plotting the expectation of the difference in the inventories with respect to the ratio $\frac{\kappa_a}{\kappa_n}$ with $\kappa_n = 2 \times 10^{-3}$ leads to the following surface:

\begin{figure}[H]
    \centering
    \includegraphics[width = 10cm]{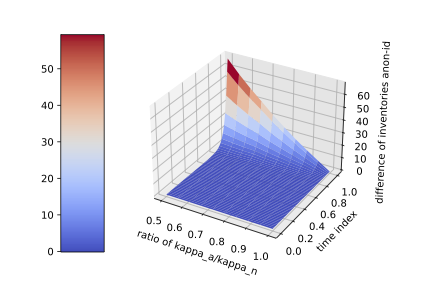}
    \caption{Sensitivity analysis of the inventory difference $Q^a_t - Q^n_t$ with respect to $\kappa_a/\kappa_n$}
    \label{fig:mf_3d_plot}
\end{figure}
Figure \ref{fig:mf_3d_plot} shows that the lower the ratio $\frac{\kappa_a}{\kappa_n}$, the bigger the difference in the inventories grows over time. In our case, the broker will accumulate more inventory through anonymous trading process than with their identity. It therefore corroborates what was obtained theoretically, namely the dependance of the optimal trading rates (and hence the inventories) on the temporary price impacts.

\subsubsection{Sensitivity to the running penalty}

We study the sensitivity of the optimal controls and inventories with respect to the running penalty by choosing $\phi \in \{10,0.01\}$ in the set of parameters \eqref{eq:params_simulation_solo_set_0}. 

\begin{figure}[H]
    \centering
    \begin{subfigure}[t]{0.45\textwidth}
        \centering
        \includegraphics[width = 7cm]{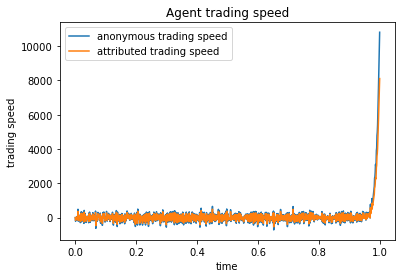}
        \caption{Optimal trading speeds}
    \end{subfigure}%
    ~ 
    \begin{subfigure}[t]{0.45\textwidth}
        \centering
          \includegraphics[width = 7cm]{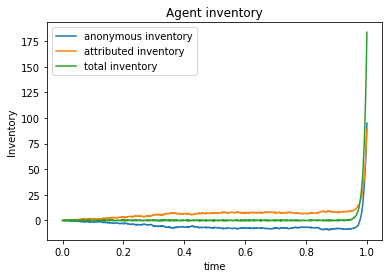}
    \caption{Optimal inventories}
    \end{subfigure}
    \caption{Behavior of the representative agent with $\phi = 10$ in \eqref{eq:params_simulation_solo_set_0}}
\end{figure}

A high running penalty confirms the intuition that the inventory held throughout the trading period should be close to zero, and that there is an overshoot at the end of the trading period to match the terminal inventory condition. Note that this is characteristic of the solutions of backward Riccati type equations. Relaxing the running penalty enables deviation from the final target in order to benefit from the opportunities brought by the difference in the temporary price impacts. The following figures in \ref{fig:small_running} illustrate this phenomenon, where there is a progressive accumulation of position in the stock until the middle of the trading period, before slowly liquidating the excess position. 

\begin{figure}[H]
    \centering
    \begin{subfigure}[t]{0.45\textwidth}
        \centering
        \includegraphics[width = 7cm]{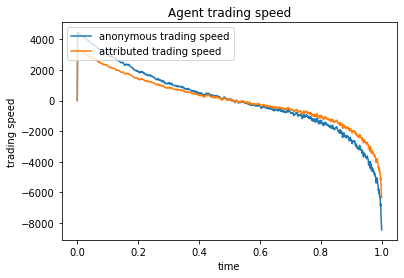}
        \caption{Optimal trading speeds}
            \label{fig:agent_trading_speeds_small_running}
    \end{subfigure}%
    ~ 
    \begin{subfigure}[t]{0.45\textwidth}
        \centering
          \includegraphics[width = 7cm]{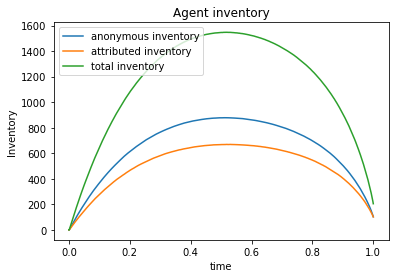}
    \caption{Optimal inventories}
    \label{fig:agent_inventories_small_running}
    \end{subfigure}
    \caption{Behavior of the representative agent with $\phi = 0.1$ in \eqref{eq:params_simulation_solo_set_0}}
    \label{fig:small_running}
\end{figure}
 
\subsubsection{Sensitivity to the terminal penalty}
We study the sensitivity of the optimal controls and inventories with respect to the terminal penalty by choosing $\Psi = 0.01$ in the set of parameters \eqref{eq:params_simulation_solo_set_0}. 

\begin{figure}[H]
    \centering
    \begin{subfigure}[t]{0.45\textwidth}
        \centering
        \includegraphics[width = 7cm]{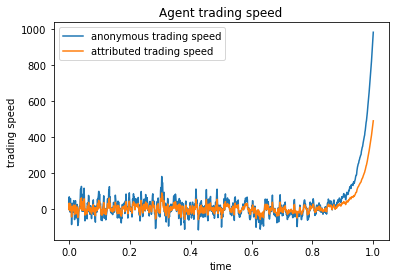}
        \caption{Optimal trading speeds}
            \label{fig:agent_trading_speeds_set_4}
    \end{subfigure}%
    ~ 
    \begin{subfigure}[t]{0.45\textwidth}
        \centering
          \includegraphics[width = 7cm]{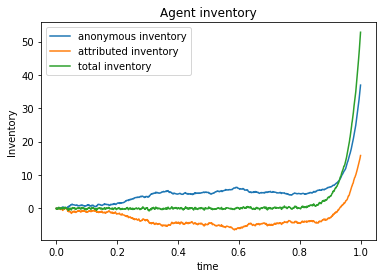}
    \caption{Optimal inventories}
    \label{fig:agent_inventories_set_4}
    \end{subfigure}
    \caption{Behavior of the representative agent for $\Psi = 0.01$ in \eqref{eq:params_simulation_solo_set_0}}
\end{figure}

The result of a small terminal penalty confirms the intuition that the terminal inventory is now allowed to be far from its target. Indeed, when the coefficient of the terminal penalty is small, the agent is willing to pay a small liquidation premium to buy the missing inventory, hence the lower inventory compared to the target of $q_T = 200$. 

\subsection{The Turnpike Effect} 

In this subsection, we highlight the \textit{turnpike effect} that is present in our model. Here we use the set of parameters 
\begin{equation}
\boxed{
\begin{array}{rcl}
& \Psi = 1, \phi \in \{1, 10\} & \text{terminal penalization and risk aversion} \\
& q_T \in \{200, 0\} & \text{target inventory}\\
& T = 1.0  & \text{time horizon} 
\end{array}
}
\label{eq:params_simulation_solo_set_turnpike}
\end{equation}
which will be modified when needed. 
\begin{figure}[H]
    \centering
    \begin{subfigure}[t]{0.45\textwidth}
        \centering
        \includegraphics[width = 7cm]{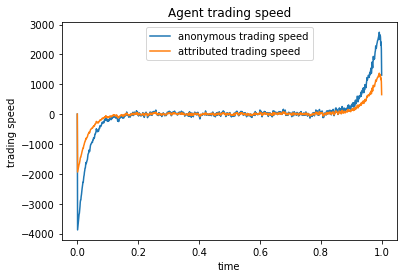}
        \caption{Optimal trading speeds}
    \end{subfigure}%
    ~ 
    \begin{subfigure}[t]{0.45\textwidth}
        \centering
          \includegraphics[width = 7cm]{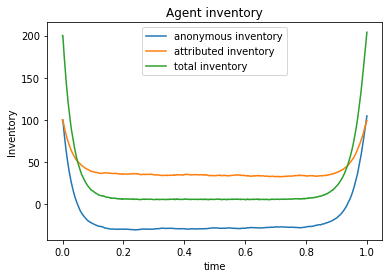}
    \caption{Optimal inventories}
    \end{subfigure}
    \caption{Illustration of the \textit{turnpike effect} for an agent's inventories and optimal controls for $\phi = 1$ in  \eqref{eq:params_simulation_solo_set_turnpike}}
\end{figure}
 
For this set of parameters, the turnpike effect is well-noticeable. When the terminal penalty is high enough, it acts as if both the initial state and the final state are\emph{essentially prescribed}. The optimal strategy of an agent is split into three pieces: the first and the last pieces are transient and exponentially short in time. They represent the transfer of the agent’s position in the stock from the initial state to the turnpike and from the latter to the final state. The middle part is a relatively long in time. It represents an efficient time evolution with minimal cost. It is what we call the turnpike. The transfer of an agent into and from the turnpike is sharper with the growth of parameter $\phi$ (as illustrated below), which is intuitive since $\phi$ controls the running penalty. Note that here we chose $Q^{\text{total}}_0 = q_T$ in order to make the different phases easier to see, but this turnpike effect is still present when the initial and final conditions are different.  
\begin{figure}[H]
    \centering
    \begin{subfigure}[t]{0.45\textwidth}
        \centering
        \includegraphics[width = 7cm]{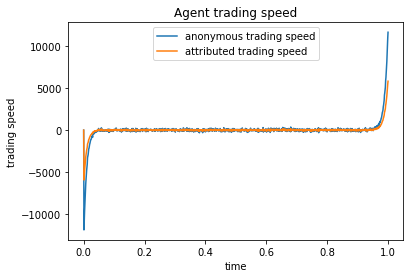}
        \caption{Optimal trading speeds}
    \end{subfigure}%
    ~ 
    \begin{subfigure}[t]{0.45\textwidth}
        \centering
          \includegraphics[width = 7cm]{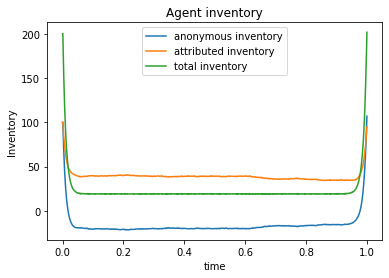}
    \caption{Optimal inventories}
    \end{subfigure}
    \caption{Illustration of the \textit{turnpike effect} for an agent's inventories and optimal controls $\phi = 10$ in  \eqref{eq:params_simulation_solo_set_turnpike}}
\end{figure}

\begin{figure}[H]
    \centering
    \includegraphics[width = 10cm]{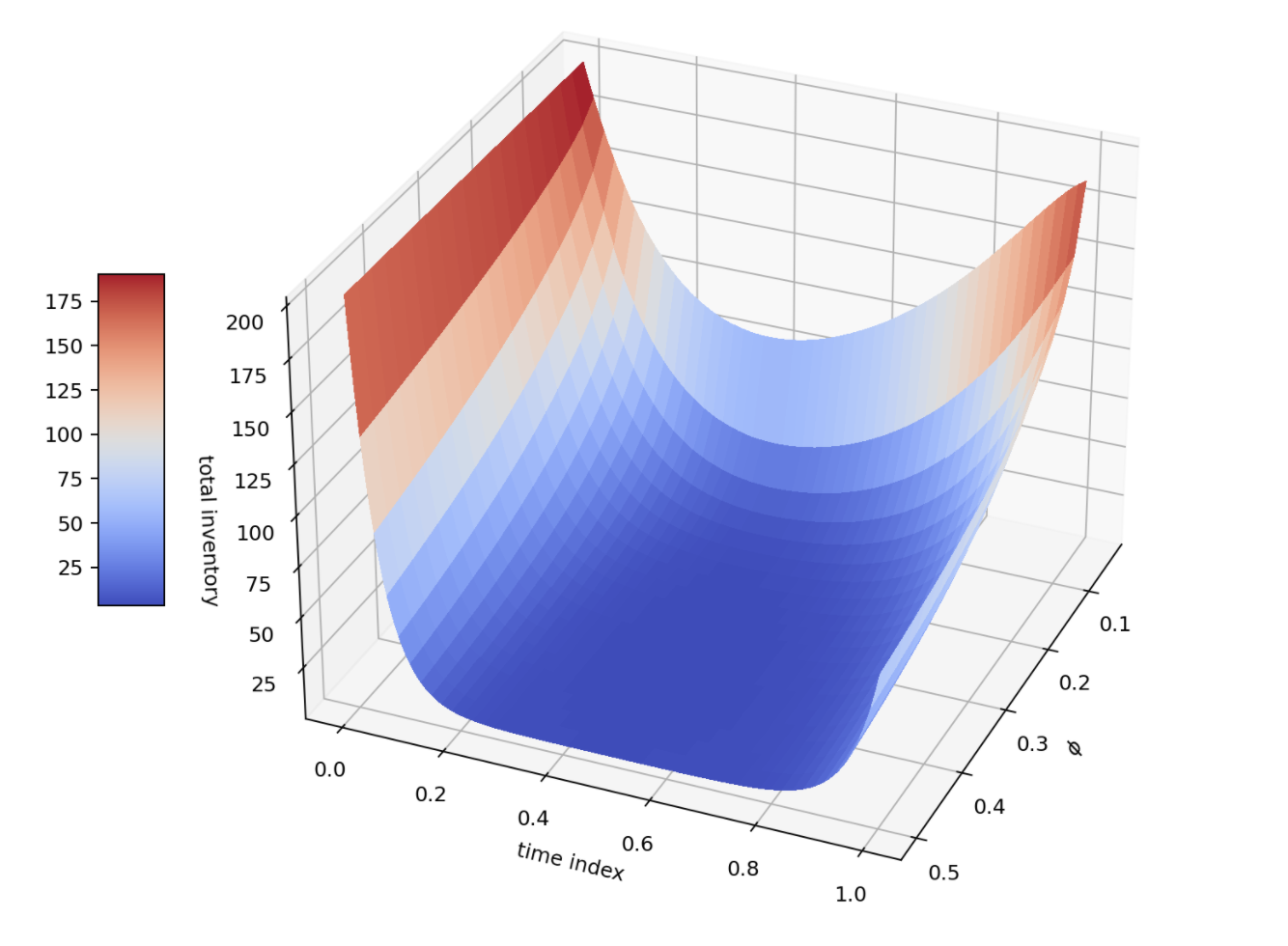}
    \caption{Sensitivity analysis of the total inventory $Q^a_t + Q^n_t$ with respect to $\phi$}
\end{figure}
For the chosen set of parameters, the critical value to observe a ``real'' turnpike phenomenon is around $\phi = 0.1$, all else being equal in \eqref{eq:params_simulation_solo_set_turnpike}.  Indeed, below this value, the turnpike effect dissipates, in the sense that the middle phase is non-existent: the two transient phases are successive, as illustrated in Figure  \ref{fig:successive-transient-turnpike}. Another instance when the turnpike effect disappears is the case when the initial and terminal inventory conditions are zero and the running cost is big enough. There are no distinct transition phases but only the central one (see Figure \ref{fig:main-pahse-turnpike}), which mainly indicates that the initial conditions and the target are on the turnpike.  

\begin{figure}[H]
    \centering
    \begin{subfigure}[t]{0.45\textwidth}
        \centering
        \includegraphics[width = 7.7cm]{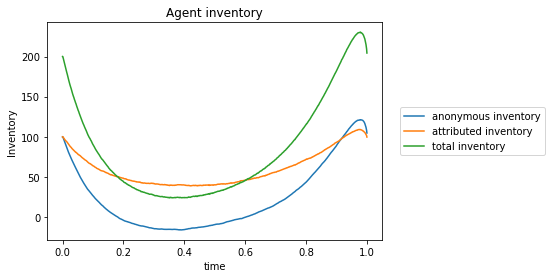}
        \caption{$Q_0^a = 100$, $Q_0^n = 100$, $q_T = 100$ and $\phi = 0.099$}
        \label{fig:successive-transient-turnpike}
    \end{subfigure}%
    ~ 
    \begin{subfigure}[t]{0.45\textwidth}
        \centering
          \includegraphics[width = 7.7cm]{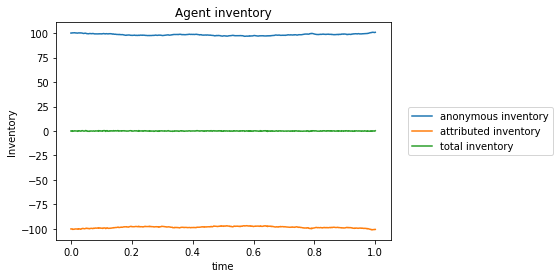}
    \caption{$Q_0^a = 100$, $Q_0^n = -100$, $q_T = 0$ and $\phi = 1$}
    \label{fig:main-pahse-turnpike}
    \end{subfigure}
    \caption{Illustration of the disappearance of the \textit{turnpike effect} with \eqref{eq:params_simulation_solo_set_turnpike}}
\end{figure}

\section{Conclusion}

In this work, we looked at a new model of optimal execution  with common noise in which the traders have the option to trade anonymously or reveal their identity. We formulated the finite-player stochastic differential game, where the interaction occurs through the impact that all traders have on the drift of the mid-price through both anonymous and identity-revealed trading rates. In the limit of a large number of agents, the interaction between the agents is captured by the conditional means of their controls. We then formulated and solved the limiting mean field limit problem before considering a specific model with Almgren-Chriss price impacts for which we show that the mean field solution provides  approximate Nash equilibriums for the finite player models. In this particular model, a representative agent uses optimal controls based on their current private inventory state $(Q^a_t, Q^n_t)$  at time $t$ and the information brought by the common noise $\mathcal{F}^0_t$ through the conditional mean of the total inventories, $\mathbb{E}^0[Q^a_t + Q^n_t]$, which represents the limit of the empirical mean of the inventories of all the brokers when their number tends to infinity. For the specific model, we derived and showed numerically that the strategic use of anonymity and identity stems from a difference in the temporary price impacts. This theoretically supports the intuition that the simultaneous empirical activity of both the anonymous and attributed trading on the Toronto Stock Exchange can be explained by a difference in execution costs. 

\begin{appendix}
\section{Proof of Theorem \ref{thm:epsilon-NE}} \label{Appendix_A}

Let us denote $\overline{\nu} = (\overline{\nu}^a, \overline{\nu}^n)$ with $\overline{\nu}^a_t = \underset{N \rightarrow \infty}{\lim} \frac{1}{N}\sum_{i=1}^N \nu^{i,a}_t$ and $\overline{\nu}^n_t = \underset{N \rightarrow \infty}{\lim} \frac{1}{N}\sum_{i=1}^N \nu^{i,n}_t$. 
By using the representation after Integration by Parts formula of the cost functionals, we obtain :
\begin{align}
    & \Big\lvert J^j(\nu, \hat{\nu}^{-j}) - J^j(\nu, \bar{\nu}) \Big\rvert =  \Big\lvert \mathbb{E}[\int_0^T (Q^{j,a}_t + Q^{j,n}_t) d(S^{N,\bar{\nu}}_t - S^{\bar{\nu}}_t)]  \Big\rvert = \Big\lvert \mathbb{E}[(Q^{j,a}_T + Q^{j,n}_T) (S^{N,\bar{\nu}}_T - S^{\bar{\nu}}_T)]  \Big\rvert
    \label{eq:difference_of_functionals} 
\end{align}
By using the definition of $S^{N,\bar{\nu}}_t$ , $S^{\bar{\nu}}_t$ and splitting the average of the controls into the $j$-th individual one and the rest as 
\begin{align*}
   & \frac{1}{N} \nu^a_t + \frac{1}{N}\sum_{k \in [N] \, \text{s.t.} k \neq j} \hat{\nu}^{k,a}_t - \bar{\nu}^a_t = \frac{1}{N} (\nu^a_t - \hat{\nu}^a_t)+ \frac{1}{N}\sum_{k \in [N]} \hat{\nu}^{k,a}_t - \bar{\nu}^a_t  \\
   & \frac{1}{N} \nu^n_t + \frac{1}{N}\sum_{k \in [N] \, \text{s.t.} k \neq j} \hat{\nu}^{k,n}_t  - \bar{\nu}^n_t  = \frac{1}{N} (\nu^n_t -\hat{ \nu}^n_t ) + \frac{1}{N}\sum_{k \in [N]} \hat{\nu}^{k}_t  - \bar{\nu}^n_t
\end{align*}
and defining $\Sigma_t =  \alpha_a \Big( \frac{1}{N} (\nu^a_t - \hat{\nu}^a_t)+ \frac{1}{N}\sum_{k \in [N]} \hat{\nu}^{k,a}_t - \bar{\nu}^a_t \Big)+\alpha_n \Big( \frac{1}{N} (\nu^n_t -\hat{ \nu}^n_t ) + \frac{1}{N}\sum_{k \in [N]} \hat{\nu}^{k}_t  - \bar{\nu}^n_t \Big)$, we obtain: 
\begin{align}
    \eqref{eq:difference_of_functionals}
        & \leq   \mathbb{E} \Big[(Q^{j,a}_T + Q^{j,n}_T)^2\Big]^{1/2} \mathbb{E} \Big[ \int_0^T \Sigma_t^2  dt \Big]^{1/2}    \quad \text{ by Cauchy-Schwarz }
\end{align}

Let us show that there exists a constant $K_j$ such that $\underset{t\in[0,T]}{\sup} \mathbb{E} \Big[(Q^{j,a}_t + Q^{j,n}_t)^2\Big] < K_j$. Let $t\in [0,T]$. 
\begin{align*}
    \mathbb{E} \Big[(Q^{j,a}_t + Q^{j,n}_t)^2\Big] & = \mathbb{E} \Big[(\int_0^t \nu^{j,a}_s ds + \int_0^t \nu^{j,n}_s ds)^2\Big]  \\
    & \leq 2 \mathbb{E} \Big[(\int_0^t \nu^{j,a}_s ds)^2 + (\int_0^t \nu^{j,n}_s ds)^2\Big] \\
     & \leq 2 \mathbb{E} \Big[\int_0^t (\nu^{j,a}_s )^2 ds + \int_0^t (\nu^{j,n}_s)^2 ds \Big] \quad \text{ by Jensen } \\
    & \leq 2 \mathbb{E} \Big[\int_0^T (\nu^{j,a}_s )^2 ds \Big] + 2\mathbb{E} \Big[\int_0^T (\nu^{j,n}_s)^2 ds \Big]  < \infty  
\end{align*}
since the controls are admissible. Since the last upper bound does not depend on $t$, there exists a constant $K_j$ independent of $t$ such that $$\mathbb{E} \Big[(Q^{j,a}_t + Q^{j,n}_t)^2\Big] < K_j$$
Hence the claim. 

\noindent Going back to \eqref{eq:difference_of_functionals} and using the claim, we obtain : 
\begin{align}
    \eqref{eq:difference_of_functionals} & \leq K_j^{1/2} \mathbb{E} \Big[\int_0^T  \Sigma_t^2 dt \Big]^{1/2} 
\end{align}
Let us work on $\mathbb{E} \Big[ \int_0^T\Sigma_t^2  dt\Big]$. By definition of $\Sigma_t$ and by using a classical inequality : 
\begin{align}
    \mathbb{E} \Big[ \int_0^T\Sigma_t^2  dt\Big] & \leq 4 \mathbb{E} \Big[ \int_0^T \Big\{  \frac{\alpha_a^2 }{N^2} (\nu^a_t - \hat{\nu}^a_t)^2+ \alpha_a^2(\frac{1}{N}\sum_{k \in [N]} \hat{\nu}^{k,a}_t - \bar{\nu}^a_t)^2 \nonumber \\
    & \qquad \qquad \quad + \frac{\alpha_n^2}{N^2} (\nu^n_t -\hat{ \nu}^n_t )^2 +\alpha_n^2 (\frac{1}{N}\sum_{k \in [N]} \hat{\nu}^{k}_t  - \bar{\nu}^n_t)^2 \Big\} dt\Big]  \nonumber \\
    & \leq 4 \max(\alpha_a^2, \alpha_n^2) \cdot \Big\{ \frac{1}{N^2} \Big(  \mathbb{E} \Big[ \int_0^T (\nu^a_t - \hat{\nu}^a_t)^2 dt \Big] +  \mathbb{E} \Big[\int_0^T (\nu^n_t -\hat{ \nu}^n_t )^2 dt\Big]\Big) \nonumber \\
    & \qquad \qquad \mathbb{E} \Big[\int_0^T (\frac{1}{N}\sum_{k \in [N]} \hat{\nu}^{k,a}_t - \bar{\nu}^a_t)^2 dt \Big] + \mathbb{E} \Big[ \int_0^T (\frac{1}{N}\sum_{k \in [N]} \hat{\nu}^{k}_t  - \bar{\nu}^n_t)^2 dt \Big] \Big\} 
\end{align}

Since $\nu^a, \hat{\nu}^a$ and $\nu^n, \hat{\nu}^n$ are in $\mathbb{H}^2$, there exists a constant $L$ such that 
\begin{equation}
    \mathbb{E} \Big[ \int_0^T (\nu^a_t - \hat{\nu}^a_t)^2 dt \Big] +  \mathbb{E} \Big[\int_0^T (\nu^n_t -\hat{ \nu}^n_t )^2 dt\Big] < L
    \label{eq:inequality_1}
\end{equation}

Let us define $D(r) := \mathbb{E} \Big[\int_0^r (\frac{1}{N}\sum_{k \in [N]} \hat{\nu}^{k,a}_t - \bar{\nu}^a_t)^2 dt \Big], r \in[0,T]$. Let $r \in [0,T]$.
\begin{align*}
    D(r) & =   \int_0^r \mathbb{E} \Big[(\frac{1}{N}\sum_{k \in [N]} \hat{\nu}^{k,a}_t - \bar{\nu}^a_t)^2 \Big] dt \text{ by Fubini-Tonelli} \\
    & =   \int_0^r \mathbb{E} \Big[(\frac{1}{N}\sum_{k \in [N]} \Big\{ \frac{-1}{2 \kappa_a} \varphi(t) (\hat{Q}^{k,a}_t + \hat{Q}^{k,n}_t - \mathbb{E}[ \hat{Q}^{k,a}_t + \hat{Q}^{k,n}_t| \mathcal{F}^0_t] )^2 \Big\} \Big] dt \\
    & =   \int_0^r \mathbb{E} \Big[\Big(\frac{1}{N}\sum_{k \in [N]} \frac{-1}{2 \kappa_a} \varphi(t) (\hat{Q}^{k,a}_t - \bar{Q}^{a}_t + \hat{Q}^{k,n}_t - \bar{Q}^{n}_t) \Big)^2 \Big] dt \\
        & \leq \frac{1}{4 \kappa_a^2}\lVert \varphi \rVert_{L^{\infty}([0,T])}   \int_0^r \mathbb{E} \Big[ \Big(\frac{1}{N}\sum_{k \in [N]} (\hat{Q}^{k,a}_t - \bar{Q}^{a}_t + \hat{Q}^{k,a}_t - \bar{Q}^{a}_t) \Big)^2 \Big] dt \\
        & \leq  \frac{1}{2 \kappa_a^2}\lVert \varphi \rVert_{L^{\infty}([0,T])}   \int_0^r \mathbb{E} \Big[ \Big(\frac{1}{N}\sum_{k \in [N]} \hat{Q}^{k,a}_t - \bar{Q}^{a}_t\Big)^2 +  \Big(\frac{1}{N}\sum_{k \in [N]} \hat{Q}^{k,a}_t - \bar{Q}^{a}_t \Big)^2 \Big] dt \\
        & \leq \frac{1}{2 \kappa_a^2} \lVert \varphi \rVert_{L^{\infty}([0,T])}   \int_0^r \mathbb{E} \Big[ \Big(\frac{1}{N}\sum_{k \in [N]} \int_0^t  \Big\{\hat{\nu}^{k,a}_s - \bar{\nu}^{a}_s \Big\}  ds \Big)^2 +  \Big(\frac{1}{N}\sum_{k \in [N]} \int_0^t  \Big\{\hat{\nu}^{k,n}_s - \bar{\nu}^{n}_s \Big\}  ds \Big)^2 \Big] dt \\ 
    & \leq \frac{1}{2 \kappa_a^2} \lVert \varphi \rVert_{L^{\infty}([0,T])}   \int_0^r \mathbb{E} \Big[ \Big( \int_0^t  \Big\{\frac{1}{N}\sum_{k \in [N]} \hat{\nu}^{k,a}_s - \bar{\nu}^{a}_s \Big\}  ds \Big)^2 +  \Big( \int_0^t \Big\{\frac{1}{N}\sum_{k \in [N]} \hat{\nu}^{k,n}_s - \bar{\nu}^{n}_s \Big\}  ds \Big)^2 \Big] dt \\ 
       & \leq \frac{1}{\kappa_a^2} \lVert \varphi \rVert_{L^{\infty}([0,T])}   \int_0^r \mathbb{E} \Big[ \Big( \int_0^t  \Big\{\frac{1}{N}\sum_{k \in [N]} \hat{\nu}^{k,a}_s - \bar{\nu}^{a}_s \Big\}^2  ds \Big) +  \Big( \int_0^t \Big\{\frac{1}{N}\sum_{k \in [N]} \hat{\nu}^{k,n}_s - \bar{\nu}^{n}_s \Big\}^2  ds \Big) \Big] dt
\end{align*}
A similar inequality holds for $E(r) = \mathbb{E} \Big[\int_0^r (\frac{1}{N}\sum_{k \in [N]} \hat{\nu}^{k,n}_t - \bar{\nu}^n_t)^2 dt \Big]$ with $\kappa_a$ replaced by $\kappa_n$. By summing those two inequalities and by Gronwall's lemma, we conclude that $D(r) + E(r) = 0$. Combining with \eqref{eq:inequality_1}, we obtain that $\eqref{eq:difference_of_functionals} = O(\frac{1}{N})$

Then, it is straightforward to deduce that for all $\omega \in  \mathcal{A}^j$: 
\begin{align}
    J^j(\omega, \hat{\nu}^{-j}) \leq \underset{\boldsymbol{\nu} \in \mathcal{A}^j}{\sup} J^j(\nu, \hat{\nu}^{-j})
\end{align}
so for $\omega = \hat{\nu}^j$, we obtain the LHS. 

From the previous result:
\begin{align*}
    J^j(\nu, \hat{\nu}^{-j}) & \leq J^j(\nu, \overline{\nu}) + O(\frac{1}{N}) \\
    &  \leq \underset{\boldsymbol{\nu} \in \mathcal{A}^j}{\sup} J^j(\nu, \overline{\nu}) + O(\frac{1}{N}) \\
    & = J^j(\hat{\nu}^j, \overline{\nu}) + O(\frac{1}{N}) \text{ by definition of } \hat{\nu}^j
\end{align*}
By taking the supremum on $\nu$, we have : 
\begin{align}
   \underset{\boldsymbol{\nu} \in \mathcal{A}^j}{\sup} J^j(\nu, \hat{\nu}^{-j}) &  \leq J^j(\hat{\nu}^j, \overline{\nu}) + O(\frac{1}{N}) \leq J^j(\hat{\nu}^j, \hat{\nu}^{-j}) + O(\frac{1}{N})
\end{align}
by applying the intermediate result at the end. Hence the RHS. 

\end{appendix}

\nocite{*}
\printbibliography

\end{document}